\documentclass[11pt]{article}
\usepackage{enumerate}
\usepackage{amsmath}
\usepackage{amssymb}
\usepackage{amsfonts}
\usepackage{graphicx}
\usepackage{fullpage}
\usepackage{algorithm}
\usepackage{wrapfig}
\usepackage[noend]{algpseudocode}
\usepackage{hyperref}
\hypersetup{pagebackref=true,colorlinks,linkcolor=red,filecolor=red,citecolor=red,urlcolor=red}
\usepackage{comment}
\excludecomment{ignore}
\excludecomment{NIIS}

\usepackage{amsthm}
\newtheorem*{definition*}{Definition}
\newtheorem*{theorem*}{Theorem}
\newtheorem{theorem}{Theorem}[section]
\newtheorem{lemma}[theorem]{Lemma}
\newtheorem{proposition}[theorem]{Proposition}

\newtheorem{observation}[theorem]{Observation}

\newtheorem{corollary}[theorem]{Corollary}

\makeatletter
\newtheorem*{rep@theorem}{\rep@title}
\newcommand{\newreptheorem}[2]{%
	\newenvironment{rep#1}[1]{%
		\def\rep@title{#2 \ref{##1}}%
		\begin{rep@theorem}}%
		{\end{rep@theorem}}}
\makeatother

\newreptheorem{theorem}{Theorem}
\newreptheorem{lemma}{Lemma}
\newreptheorem{claim}{Claim}
\newreptheorem{corollary}{Corollary}

\newcommand{\namedref}[2]{\hyperref[#2]{#1~\ref*{#2}}}

\newcommand{\sectionref}[1]{\namedref{Section}{#1}}

\newcommand{\observationref}[1]{\namedref{Observation}{#1}}

\newcommand{\figureref}[1]{\namedref{Figure}{#1}}
\newcommand{\figurerefb}[2]{\hyperref[#1]{Figure~\ref*{#1}#2}}

\newcommand{\lemmaref}[1]{\namedref{Lemma}{#1}}

\newcommand{\propositionref}[1]{\namedref{Proposition}{#1}}
\newcommand{\corollaryref}[1]{\namedref{Corollary}{#1}}

\newcommand{\equationref}[1]{\hyperref[#1]{(\ref*{#1})}}
\newcommand{\id}[1]{{\tt #1}}

\newcommand{\none}{\textsc{none}}
\newcommand{\config}[1]{{\mathcal{#1}}}

\begin{document}

\title{Why Extension-Based Proofs Fail}
\date{\today}

\author{
Dan Alistarh\\
  IST Austria\\
  dan.alistarh@ist.ac.at\\
  \and
James Aspnes\\
  Yale\\
  james.aspnes@gmail.com\\
  \and
Faith Ellen\\
  University of Toronto\\
  faith@cs.toronto.edu\\
  \and
Rati Gelashvili\\
  University of Toronto\\
  gelash@gmail.com\\
  \and
Leqi Zhu\\
  University of Michigan\\
  zhu.jimmy@gmail.com\
  }

\maketitle

\begin{abstract}
We introduce \emph{extension-based proofs},  a class of impossibility proofs that includes valency arguments. They are modelled as an interaction between a prover and a protocol.
Using proofs based on combinatorial topology, it has been shown that
it is impossible to deterministically solve $k$-set agreement among $n > k \geq 2$ processes in a wait-free manner in
certain asynchronous models.
However, it was unknown whether proofs based on simpler techniques were possible.  We show that this impossibility result cannot be obtained 
for one of these models
by an extension-based proof and, hence, extension-based proofs are limited in power.
\end{abstract}
\section{Introduction}

One of the most well-known results in the theory of distributed computing, due to Fischer, Lynch, and Paterson~\cite{FLP85}, is that there is no deterministic, wait-free protocol solving consensus among $n \geq 2$ processes in an asynchronous message passing system, even if at most one process may crash.
Their result has been extended to asynchronous shared memory systems where processes communicate by reading from and writing to shared registers \cite{Abr88,CIL87,Her91,LA87}.
Moses and Rajsbaum~\cite{MR02} gave a unified framework for proving the impossibility of consensus in a number of different systems.

Chaudhuri~\cite{Cha93} conjectured that 
the impossibility of consensus
could be generalized to the \emph{$k$-set agreement problem}. In this problem, there are $n > k \geq 1$ processes, each starting with an input in $\{0,1,\dots,k\}$.
Each process that does not crash must output a value that is the input of some process
(\emph{validity}) and, collectively, at most $k$ different values may be output (\emph{agreement}).
In particular, \emph{consensus} is just 1-set agreement.

Chaudhuri's conjecture was eventually proved in three concurrent papers by Borowsky and Gafni~\cite{BG93a}, Herlihy and Shavit~\cite{HS99}, and Saks and Zaharoglou~\cite{SZ00}. 
These proofs
and a later proof by Attiya and Rajsbaum~\cite{AR02} all relied on sophisticated machinery from combinatorial topology, using  a simplicial complex 
to model the set of all initial configurations of a wait-free protocol and a subdivision of it to model the set of all its final configurations.
Then they used Sperner's Lemma to show that there exists a final configuration in which $k+1$ different values have been output.
This proves that the protocol does not correctly solve $k$-set agreement.

Later on, Attiya and Casta\~{n}eda~\cite{AC11} and Attiya and Paz~\cite{AP12} showed how to obtain the same results using purely combinatorial techniques, without explicitly using topology.
Like the topological proofs, these proofs also consider the set of final configurations of a supposedly wait-free $k$-set agreement protocol. However, 
by relating different final configurations to one another using indistinguishability and employing arguments similar to proofs of Sperner's Lemma, they
proved the existence of a final configuration in which $k+1$ different values have been output.
 
A common feature of 
these impossibility proofs
is that they are non-constructive.
They 
prove that any deterministic protocol for $k$-set agreement 
among $n > k$ processes in an asynchronous system
has an execution in which some process takes infinitely many steps  without returning a value, but do not construct such an execution. 

In contrast,  impossibility proofs for deterministic, wait-free consensus in asynchronous systems
explicitly construct an infinite execution by repeatedly extending a finite execution by the steps of some processes.
Specifically, they define a {\em bivalent configuration} to be a configuration from which there is an execution in which some process outputs 0 and an execution in which some process outputs 1. Then they show that, from any bivalent configuration, there is a step of some process that results in another bivalent configuration. This allows them to explicitly construct an infinite execution in which no process has output a value.
A natural question arises: is there a
proof of the impossibility of
$k$-set agreement that explicitly constructs an infinite execution by repeated extensions? 
This question is related to results in proof complexity that show certain theorems cannot be obtained in 
weak formal
systems. For example, it is known that relativized bounded arithmetic cannot prove the pigeonhole principle~\cite{PBI93}.

\emph{Our contributions.} 
In this paper, we formally
define the class of \emph{extension-based proofs},
which model impossibility proofs that explicitly construct an infinite execution by repeated extensions.
We also prove that there is no extension-based proof of the impossibility of a deterministic, wait-free protocol solving $k$-set agreement among $n> k \geq 2$ processes in asynchronous systems where processes communicate
using an unbounded sequence of snapshot objects, to which each process can \id{update} and \id{scan} only once.

A \emph{task} is a problem in which each process starts with a private input value and must output one value, such that the sequence of values produced by the processes satisfies certain specifications, which may depend on the input values of the processes.
We view a proof of the impossibility of solving a task as an interaction between a \emph{prover} 
and any protocol that claims to solve the task. 
The prover has to refute this claim.
To do so, it can repeatedly query
the protocol about the states of processes in configurations that can be reached in a small
number of steps from configurations it already knows about.
It can also ask the protocol to exhibit an execution by a set of processes from a configuration it knows about in which some process outputs a particular value, or to declare that no such execution exists.
The goal of the prover is to construct a \emph{bad} execution, i.e.~an execution in which some processes take infinitely many steps 
without terminating or output values that
do not satisfy the specifications of the task.
The definition of extension-based proofs is presented in \sectionref{sec:ebp}.

A key observation
is that, from the results of its queries, many protocols are indistinguishable
to the prover. 
It must construct a single execution that is bad for all these protocols.
To prove that no prover can 
construct a bad execution,
we show how an adversary 
can adaptively
define a protocol in response to any specific prover's queries.
In this \emph{adversarial protocol,} 
all processes eventually terminate and output correct values
in executions consistent with the results of the prover's queries.
In~\sectionref{sec:ext-extensions}, we argue that no
such proof can
refute the possibility of a deterministic, wait-free protocol solving $k$-set agreement among $n > k \geq 2$ processes 
in the  {\em non-uniform iterated snapshot} (NIS) model.
In the conference version of this paper~\cite{AAEGZ19}, we made a very similar argument in the {\em non-uniform iterated immediate snapshot} (NIIS) model.
The snapshot and NIS models are defined in \sectionref{sec:model}.

From a computability standpoint, the snapshot, NIS, and NIIS models are equivalent in power to the basic asynchronous model
in which processes communicate through shared registers:
Any protocol in the basic asynchronous model can be easily adapted to run in the snapshot model
by replacing each \id{read} by a \id{scan} and then throwing away the information it does not need.
Afek, Attiya, Dolev, Gafni, Merritt, and Shavit~\cite{AADGMS93} gave a wait-free implementation of a snapshot object
using registers, so the converse is also true.
Any execution of a protocol using an immediate snapshot object is also an execution of the protocol using a snapshot object, so protocols designed for the
snapshot model also run in the immediate snapshot model.
Borowsky and Gafni~\cite{BG93b}
gave a wait-free implementation of an immediate snapshot object from a snapshot object, so protocols designed for
the immediate snapshot model can be modified to run in the snapshot model.
Likewise,  the non-uniform iterated snapshot (NIS) model and the non-uniform iterated immediate snapshot (NIIS) model
are computationally equivalent.
The {\em iterated immediate snapshot}
(IIS) model was introduced by Borowsky and Gafni~\cite{BG97}.
The NIIS model is a slight generalization, introduced by Hoest and Shavit~\cite{HS2006}.
Any protocol in the NIS model can be easily adapted to run in the single-writer snapshot model by appending values
rather than overwriting them when performing an \id{update} and throwing away the information in a \id{scan} about values that would have appeared
in other snapshot objects.
Similarly, any protocol in the NIIS model can be adapted to run in the immediate snapshot model.
Borowsky and Gafni gave a nonblocking emulation in the NIIS model of any protocol in the immediate snapshot model~\cite{BG97}. The same emulation can be applied to a protocol in the snapshot model to obtain a protocol in the NIS model.
Thus, to show there is no bounded wait-free protocol to solve a certain task
in all of these models, it suffices to show that there is no bounded wait-free protocol to solve that task in any one of them.

The IIS and NIIS models are nice because the reachable configurations of a protocol have a natural representation
using combinatorial topology.
For the NIS model, there is a similar representation using 
graphs, which may be easier
to understand.
This representation is presented 
 in \sectionref{sec:prelim}, together with some properties
that are needed for our proof.
In this view, when an extension-based prover makes queries, it is essentially performing local search on the configuration space of the protocol. Because the prover 
obtains incomplete information about the protocol, the adversary has some flexibility when specifying the protocol's behaviour
in configurations not yet queried by the prover.

There are a number of interesting directions for extending this work, which are discussed in~\sectionref{sec:conclusions}.
\section{Models}
\label{sec:model}

An execution is a sequence of steps. In each step of a shared memory model,
a process performs an atomic operation on a shared object and then updates its local state.
Communication among processes occurs through the atomic operations on shared objects.
We use $n$ to denote the number of processes, $p_1,\dots,p_{n}$ to denote the processes, and $x_i$ to denote the input to process $p_i$. When solving a task, we let $y_i$ denote the output of process $p_i$. A value is assigned to $y_i$ immediately
before $p_i$ terminates.

In the {\em single-writer register} model, there is one shared register $R_i$ for each process $p_i$, to which only it can \id{write}, but which can be \id{read}
by every process. The initial value of each register is $-$.

In the {\em single-writer snapshot} model, there is one shared single-writer snapshot object $S$ with $n$ components.
The initial value of each component is $-$.
The snapshot object supports two operations, $\id{update}(v)$ and $\id{scan}()$. An $\id{update}(v)$ operation by process $p_i$ updates $S[i]$, the $i$'th component of $S$,  to have value $v$, where $v$ is an element of an arbitrarily large set that does not contain $-$. A $\id{scan}()$ operation returns the value of each component of $S$.

In the {\em non-uniform iterated snapshot} (NIS) model, there is an infinite sequence, $S_1,S_2,\dots$, of shared \emph{single-writer snapshot} objects, each with $n$ components.
The initial value of each component is $-$.
The \emph{initial state} of process $p_i$
consists of its identifier, $i$, and its input, $x_i$.
Each process accesses each snapshot object at most twice, starting with $S_1$. 
The first time $p_i$ accesses a snapshot object $S_r$, it performs $\id{update}(s_i)$ to set  $S_r[i]$ to its current state, $s_i$.
Its new state is the same as its previous state, except for an extra bit indicating that it has performed the \id{update}.
At its next step, it performs a $\id{scan}$ of $S_r$.
Its new state, $s'_i$, is a pair consisting of $i$ and the result of the $\id{scan}$.
Process $p_i$ remembers its entire history, because $s_i$ is the $i$'th component of the result of the $\id{scan}$.
Next, $p_i$  consults a function, $\delta$, from the set of possible states of processes to the set of possible output values
and the special symbol $\bot$.
This indicates
whether $p_i$ should output a value: If $\delta(s'_i) \neq \bot$, then $p_i$ outputs $\delta(s'_i)$ and 
is {\em terminated}.
If $\delta(s'_i) =\bot$, then $p_i$ is poised to access the next snapshot object.
A protocol in the NIS model is completely specified by the function $\delta$. 

In the snapshot and NIS models, we assume that a process is only terminated after performing a $\id{scan}$. This is without loss of generality,
because a $\id{scan}$ does not change the contents of shared memory and, so, does not affect any other process.

A \emph{configuration} of a protocol consists of the contents of each shared object and the state of each process at some point during an execution of the protocol. 
An \emph{initial} configuration is a configuration in which 
every process is in an initial state and every object has its initial value.
A process is \emph{active} in a configuration if it is not terminated.
A configuration is \emph{final} if it has no active processes.
If $C$ is a configuration and $p_i$ is a process that is active in $C$, then $Cp_i$ denotes the configuration that results
when $p_i$ takes the step from configuration $C$ specified by the protocol.
A \emph{schedule} from $C$ is a finite or infinite sequence of (not necessarily distinct) processes $\alpha = \alpha_1,\alpha_2,\ldots$ such that there is a sequence of configurations $C= C_0,C_1,C_2\ldots$
where process $\alpha_i$ is active in $C_{i-1}$ and $C_i = C_{i-1}\alpha_i$ for each process $\alpha_i$ in $\alpha$.
If $\alpha$ is a finite schedule from $C$, then $C\alpha$ denotes the configuration of the protocol that is reached by performing steps, 
with one step by a process for each occurrence of the process in $\alpha$, in order, starting from  configuration $C$.
Each finite schedule from an initial configuration results in a \emph{reachable} configuration.
A protocol is \emph{wait-free} if it does not have an infinite schedule.

Two configurations $C$ and $C'$ are \emph{indistinguishable} to a set of processes $P$ 
if every process in $P$ has the same state in $C$ and $C'$.
Two finite schedules $\alpha$ and $\beta$ from $C$ are \emph{indistinguishable} to the set of processes $P$ if the resulting 
configurations $C\alpha$ and $C\beta$ are indistinguishable to $P$.
A \emph{P-only schedule} from $C$ is a schedule in which only processes in $P$ appear.
\section{Extension-Based Proofs}
\label{sec:ebp}

\begin{NIIS}
A \emph{configuration} of a protocol consists of the contents of each shared object and the state of each process at some point during an execution. 
For any configuration $C$ and any finite sequence $\alpha$ of (not necessarily distinct) processes,
let $C\alpha$ denote the configuration of the protocol that is reached by performing steps of the protocol, with one step by a process for each occurrence of the process in $\alpha$, in order, starting from  configuration $C$. For example, $Cp_2p_1p_1$ is the 
configuration of the protocol reached from $C$ by performing one step by process $p_2$ and then two steps by process $p_1$.
A process is \emph{active} in a configuration if it has not terminated.
Without loss of generality, we assume that, immediately after a process outputs a value, it terminates.
A configuration is \emph{final} if it has no active processes. An \emph{initial} configuration is a configuration in which no process has taken any steps and it is specified by the input of each process.
\end{NIIS}

An \emph{extension-based} proof is an interaction between a prover and a (supposedly) wait-free protocol for solving a task, which the prover is trying to prove is incorrect.
The prover starts with no knowledge about the protocol (except its initial configurations) and makes the protocol reveal information about
various configurations by asking \emph{queries}, which it chooses adaptively, based on the responses to its queries.
The interaction proceeds in phases, beginning with phase 1.  
 
In each phase $\varphi \geq 1$, the prover starts with a finite schedule, $\alpha(\varphi)$,  from a set of initial configurations,
$\config{B}(\varphi)$, 
which only differ from one another in the input values of processes that do not occur in $\alpha(\varphi)$,
and the set of resulting configurations $\config{A}(\varphi) = \{ C_0\alpha(\varphi) \ |\ C_0 \in \config{B}(\varphi)\}$.
At the start of phase 1, $\alpha(1)$ is the empty schedule and $\config{A}(1) = \config{B}(1)$ is the set of all initial configurations of the protocol.
If every configuration in $\config{A}(\varphi)$ is final and the values output by the processes before
terminating
satisfy the specifications
of the task, then the prover \emph{loses}.
The prover also maintains a set, $\config{A}'(\varphi)$, containing the configurations it reaches 
by taking non-empty sequences of steps from configurations in  $\config{A}(\varphi)$ during phase $\varphi$.
This set is empty at the start of phase $\varphi$ and it will be constructed so that,
for every configuration $C' \in \config{A}'(\varphi)$, there exists a configuration $C \in \config{A}(\varphi)$
and a schedule $\beta$ from $C$ such that $C' = C \beta$ and $C \beta' \in \config{A}'(\varphi)$ for every nonempty
prefix $\beta'$ of $\beta$.

A \emph{query} $(C,q)$ by the prover is specified by
a configuration $C \in \config{A} (\varphi) \cup \config{A}'(\varphi)$ and 
a process $q$ that is active in $C$.
The protocol replies to this query with 
the configuration $C'$ resulting from $q$ taking the step from $C$ specified by the protocol.
Then the prover adds $C'$ to $\config{A}'(\varphi)$
and we say that the prover has \emph{reached} $C'$.
Since there exists a configuration $C_0 \in \config{A} (\varphi)$ and a schedule $\beta$ from
$C_0$ such that $C = C_0 \beta$ and $C_0\beta' \in \config{A}'(\varphi)$ for every nonempty
prefix $\beta'$ of $\beta$, the same is true for $C'$ with the schedule $\beta q$ from $C_0$.
If the prover reaches a configuration $C'$ in which the outputs of the processes do not satisfy
the specifications of the task, it has demonstrated that the protocol is incorrect.
In this case, the prover \emph{wins}.

A \emph{chain of queries} is a (finite or infinite) sequence of queries $(C_1,q_1),(C_2,q_2),\ldots$ such that, for all consecutive queries $(C_i,q_i)$ and $(C_{i+1},q_{i+1})$  in the chain, $C_{i+1}$ is the configuration resulting from $q_i$ taking 
the step  from $C_i$ specified by the protocol.

An \emph{output query} $(C,Q,y)$ in phase $\varphi$ is specified by a configuration $C \in \config{A}(\varphi) \cup \config{A}'(\varphi)$, a set of processes $Q$ that are all active in $C$,
and a possible output value $y$. 
If there is a $Q$-only schedule
starting from $C$ that results in a configuration in which some process in $Q$ outputs $y$,
then the protocol returns some such schedule.
Otherwise, the protocol returns \none.
Note that the prover does not add the resulting configuration to $\config{A}'(\varphi)$.
However, it can do so by asking a chain of queries starting from $C$ following the schedule returned by the protocol.
If $Q$ is the set of all processes, then the results of the 
output queries $(C,Q,y)$, for every possible output value $y$, tells the prover 
which values can be output by 
the protocol starting from configuration $C$.
For example, if 0 and 1 are the only possible output values, this enables
the prover to determine whether
$C$  is \emph{bivalent}.

After constructing finitely many output queries and chains of queries in phase $\varphi$ without winning, the prover must end the phase by committing to a nonempty schedule $\alpha'$ from some configuration $C \in \config{A}(\varphi)$ such that
$C\alpha' \in \config{A}'(\varphi)$.
Since
 $\config{A}(\varphi) = \{ C_0\alpha(\varphi) \ |\ C_0 \in \config{B}(\varphi)\}$,
there is an initial configuration $C'_0 \in \config{B}(\varphi)$ such that $C = C'_0 \alpha(\varphi)$.
Hence $C\alpha' = C'_0 \alpha(\varphi)\alpha'$.
Then $\alpha(\varphi+1) = \alpha(\varphi)\alpha'$,
$\config{B}(\varphi+1)$ 
is
the set of all initial configurations that
only differ from $C'_0$ by the states of processes that do not appear in this schedule,
and
$\config{A}(\varphi+1)= 
\{C_0 \alpha(\varphi+1)\ |\ C_0 \in \config{B}(\varphi+1)\}$.
Then the prover begins phase $\varphi+1$.
 
If the interaction between the prover and the protocol is infinite, either because the prover 
is allowed to continue a chain of queries indefinitely
or the number of phases is infinite, the prover \emph{wins}.
In this case, the prover has demonstrated that the protocol is not wait-free.
For example, a valency proof for the impossibility of consensus shows how to construct an infinite schedule for any protocol that satisfies agreement and validity.
More generally,
if the protocol satisfies the specifications of the task in all of its final configurations, making the interaction go on forever is the only way that the prover can win.
For the trivial protocol in which no process ever outputs a value, the prover can win by asking any infinite chain of queries.

To prove that a task is impossible using an extension-based proof, one must show there exists a prover that  wins against {\em every} protocol.

If a deterministic protocol is wait-free and there are only a finite number of initial
configurations, then, by K\"onig's lemma, there is a finite upper bound on the length of all schedules of the protocol.
If the prover is given (or is able to ask for) such an upper bound, it can perform a finite number of chains of queries to examine all reachable configurations. In other words, this allows the prover to perform exhaustive search in the first phase to learn everything about the protocol. 
Likewise, if a prover does not have to eventually end phase 1, it can win against every wait-free protocol
by performing exhaustive search.
Such proofs violate the spirit of extension-based proofs.

For any protocol in the IIS model, there is a bound $T$ such that every process terminates after taking exactly $T$ steps.
Although the prover is not given $T$, it is easy for a prover to determine $T$ by performing one output query or one chain of queries.
Thus, extension-based proofs are too powerful in the IIS model.
If a task has a finite number of initial configurations and there is a protocol to solve
this task in the NIIS model, then there is a protocol to solve this task in the IIS model.
However, without knowledge of an upper bound on the length of all schedules, there is no general way to construct an IIS protocol from an NIIS protocol.
\section{Properties of the NIS Model} 
\label{sec:prelim}
\newcommand{\gt}[1]{\mathbb{G}_{#1}}
\newcommand{\vt}[1]{\mathbb{V}_{#1}}
\newcommand{\et}[1]{\mathbb{E}_{#1}}
\newcommand{\at}{\mathbb{A}}
\newcommand{\bt}{\mathbb{B}}
\newcommand{\ct}{\mathbb{C}}
\newcommand{\ft}{\mathbb{F}}
\newcommand{\Tt}{\mathbb{T}}

The proof of our main result relies on properties of the non-uniform iterated snapshot model, 
including a simple graphical representation of protocols in this model.
We begin with two simple observations.

\begin{observation}
\label{obs:configstates}
A reachable configuration in the NIS model is completely determined by the states of all processes in the configuration (including the processes that have terminated).
\end{observation}
\noindent
This is true because only process $p_i$ can update the $i$'th component of each snapshot object and each process remembers its entire history.

\medskip

The second observation is a special case of a general, well-known result about indistinguishability. (For example, see Corollary 2.2. in \cite{AE14}.)

\begin{observation}
\label{obs:ccp}
Suppose $C$ and $C'$ are two reachable configurations in the NIS model  and
each snapshot object $S_r$ has the same contents in $C$ and $C'$, for all $r \geq t$.
If  $C$ and $C'$ are indistinguishable to a set of processes $P$, each active process in $P$ is poised in $C$ to access a snapshot object $S_r$, for some $r \geq t$,
and $\alpha$ is a finite, $P$-only schedule from $C$,
then $\alpha$ is a schedule from $C'$ and
the configurations $C\alpha$ and $C'\alpha$ are indistinguishable to $P$.
\end{observation}

Suppose $C$ is a reachable configuration in which all active processes are poised to $\id{update}$ the same snapshot object.
A \emph{1-round schedule} from $C$ is a schedule consisting of two occurrences of each process that is active in $C$.
Each active process in the resulting configuration is poised to $\id{update}$ the next snapshot object in the sequence.
If none of the processes are active in $C$, then the empty schedule is the only 1-round schedule from $C$.
The following observation is a corollary of \observationref{obs:ccp}.

\begin{observation}
\label{obs:indist-oneround}
Suppose $\beta$ is a 1-round schedule from
$C$, $\alpha$ is a prefix of $\beta$, and $P$ is the set of those processes that occur twice in $\alpha$,.
Then $\alpha$ and $\beta$ are indistinguishable to $P$
and the terminated processes in configuration $C$.
\end{observation}

For $t > 1$, a \emph{t-round schedule} from $C$  is a schedule $\beta_1 \beta_2 \cdots\beta_t$ such that $\beta_1$ is a 1-round schedule from $C$ and, for $1 < i \leq t$,  $\beta_i$ is a 1-round schedule from $C \beta_1 \cdots \beta_{i-1}$. Notice that some processes may have terminated during $\beta_1 \cdots \beta_{i-1}$. 
These processes are not included in $\beta_i$.

\begin{ignore}
Suppose $\beta_1 \cdots \beta_t$ is a $t$-round schedule from $C$ and, for $1 \leq i \leq t$, 
$\alpha_i$ is a prefix of $\beta_i$.
Let $P$ be the set of processes that occur twice in $\alpha_1$. 
If $\alpha_i$ is a $P$-only schedule from $C\alpha_1 \cdots \alpha_{i-1}$ and all processes in $P$ that are active in $C\alpha_1 \cdots \alpha_{i-1}$ occur twice in $\alpha_i$,
for $1 < i \leq t$,
then $\alpha_1 \cdots \alpha_t$ and $\beta_1 \cdots \beta_t$ are indistinguishable to the processes in $P$ and the terminated processes in $C$.
\end{ignore}

Every schedule from an initial configuration
$C$ that reaches a final configuration $C'$ is indistinguishable
(to all processes) to an $r$-round schedule, for some value of $r$.
This is a special case of the following lemma, where $t =1$ and $P$ is the set of all processes.

\begin{lemma}
\label{lem:round}
Let $C$ be a configuration in which every active process is poised to perform an \id{update} to $S_t$
and let $C'$ be a configuration reachable from $C$. Suppose that $P$ is a set of processes that 
is each poised to perform an \id{update} to $S_{t+r}$ in $C'$ or has terminated prior to performing an \id{update} to $S_{t+r}$.
Then there exists an $r$-round schedule $\beta$ from $C$
such that $C\beta$ and $C'$ are  indistinguishable to $P$, i.e.~each process in $P$ has the
same state in $C\beta$ and $C'$.
\end{lemma}

\begin{proof}
Since $C'$ is reachable from $C$, there is a finite schedule $\alpha$ from $C$ such that $C' = C\alpha$.
Note that each process in $P$ occurs at most $2r$ times in $\alpha$.
Let $\gamma$ be the schedule from $C$ obtained from $\alpha$ by removing all but the first $2r$ occurrences of every process.
The steps that are performed in $\alpha$, but not $\gamma$, are accesses to $S_{t+r}$ or snapshot objects that follow $S_{t+r}$.
Since no process in $P$ accesses these objects when the protocol is performed from $C$ according to $\alpha$
or $\gamma$, $C\alpha$ and $C\gamma$ are indistinguishable to $P$.
So, it suffices to show that $C\gamma$ and $C\beta$ are indistinguishable to $P$ for some $r$-round schedule $\beta$.

The proof proceeds by induction on $r$.
First suppose that $r =1$. Let $\beta$ be any 1-round schedule that has $\gamma$ as a prefix.
In other words, append sufficiently many occurrences of each process that is active in $C$ to the end of $\gamma$
so that each occurs exactly 2 times in $\beta$.
Note that each process in $P$ that is active in $C$ occurs exactly 2 times in $\alpha$, so all occurrences of processes in $P$
that occur in $\beta$ occur in $\gamma$. 
By \observationref{obs:indist-oneround}, $C\gamma$ and $C\beta$ are indistinguishable to $P$.

Now suppose that $r > 1$. Let $\alpha'$ be the schedule from $C$ obtained from $\gamma$ by removing all but the first $2r-2$ occurrences of every process.
This removes all accesses of $S_{t+r-1}$, but no other steps.
Let $P'$ be the set of all processes that are poised to perform an \id{update} to $S_{t+r-1}$ in $C\alpha'$ or have terminated prior to performing an \id{update} to $S_{t+r-1}$.
Then $P \subseteq P'$.
By the induction hypothesis, there exists an $(r-1)$-round schedule $\beta'$ from $C$
such that $C\beta'$ and $C\alpha'$ are indistinguishable to $P'$.

Let $\alpha''$ be the schedule from $C\alpha'$ obtained from $\gamma$ by removing the first $2r-2$ occurrences of every process.
Each process that is terminated in $C\alpha'$ does not occur in $\alpha''$.
Each process in $P$ that is active in $C\alpha'$ occurs exactly twice in $\alpha''$.

Let $\beta''$ be a 1-round schedule from $C\beta'$ obtained from $\alpha''$ by appending sufficiently many occurrences of every process that is active in $C\beta'$ so that each occurs exactly twice  in $\beta''$. 

Note that if some process $p_i$ performs its \id{update} to $S_{t+r-1}$ in $\gamma$
then $p_i \in P'$ and $p_i$ occurs at least once in $\alpha''$, so it is active in $C\alpha'$
and performs the same \id{update} to $S_{t+r-1}$ in $\alpha''$.
Since $C\beta'$ and $C\alpha'$ are indistinguishable to $p_i$, it performs the same \id{update} to $S_{t+r-1}$ in $\beta''$.
By construction, the accesses of $S_{t+r-1}$ in $\alpha''$  occur in the same order as in $\gamma$.
Moreover, because each process in $P$ that is active in $C\alpha'$ occurs exactly twice in $\alpha''$,
its \id{scan} of $S_{t+r-1}$ gets the same result in $\gamma$, $\alpha''$, and $\beta''$.
Hence $C\gamma$ and $C\beta$ are indistinguishable to $P$, where $\beta = \beta'\beta''$ is an $r$-round schedule.
\end{proof}

In particular, the state of each process in each reachable configuration in which the process is poised to perform
an \id{update} to $S_{1+r}$ 
or has terminated prior to performing an \id{update} to $S_{1+r}$
is the state of that process in a configuration reachable by an $r$-round schedule
from an initial configuration.

Consider a protocol in the NIS model (specified by a function $\delta$ from the set of possible states of processes
to the set of possible output values and the special symbol $\bot$).
We use an undirected graph $\gt{t} = (\vt{t},\et{t})$ to represent the configurations of this protocol reachable from initial configurations by $t$-round schedules.
Each vertex $v \in \vt{t}$ represents the state of one process in some such reachable configuration
and $\mathit{id}(v)$ is the identifier of this process, which is the first part of the state.
There is an edge in $\et{t}$  between two vertices if there is some such reachable configuration that contains 
the states represented by both vertices.
In each configuration, there are exactly $n$ vertices, each with a different identifier.
Therefore each edge in $\et{t}$ belongs to an $n$-vertex clique in $\gt{t}$ consisting of vertices with distinct identifiers.

An $n$-vertex clique \emph{represents} a configuration if the vertices of the clique represent the states of the processes in that configuration.
In particular, the $n$-vertex cliques in $\gt{0}$ represent all initial configurations.
For the $k$-set agreement problem, $\vt{0} = \left\{ (i,a) \ |\ i \in \{1,\ldots, n\} \mbox{ and } a \in \{0, \ldots, k\} \right\}$
and $\{(i,a),(j,b)\} \in \et{0}$ if and only if $i \neq j$.
A vertex $v$ is {\em active} if the state it represents is active and we use $\delta(v) = \bot$ to denote this.
A vertex $v$ is {\em terminated} if the state it represents is terminated and we use $\delta(v) \neq \bot$ to denote the
value the process outputs in this state.
An $n$-vertex clique represents a final configuration if and only if all its vertices are terminated.

We show how to construct $\gt{t+1}$ from $\gt{t}$,
given $\delta(v)$ for all $v \in \vt{t}$.
We start with an $n$-vertex clique $\sigma$ in $\gt{t}$, which represents some configuration $C$ reachable from an initial configuration by a $t$-round schedule,
and construct the $n$-vertex cliques of $\gt{t+1}$ representing configurations reachable from $C$ by 1-round schedules.

Consider any subset $\tau$ of the active vertices in $\sigma$.
Let $\mathit{id}(\tau) = \{ \mathit{id}(v) \ |\ v \in \tau\}$ be the set of identifiers of processes whose states are represented by
vertices in $\tau$.
Each process $p_i$, for $i \in \mathit{id}(\tau)$, is poised to 
perform an $\id{update}$ to $S_{t+1}$ in configuration $C$.
Suppose process $p_i$ performs its  $\id{update}$ to $S_{t+1}$, for each $i \in \mathit{id}(\tau)$, but no other process does so.
Then, for each $v \in \tau$, $S_{t+1}[\mathit{id}(v)]$ is the state represented by $v$
and, for each $j \not\in \mathit{id}(\tau)$, $S_{t+1}[j] = -$.
If some process $p_i$ now performs a $\id{scan}$ of $S_{t+1}$, the result is an $n$-component vector containing these values. Since there is a one-to-one correspondence between $\tau$ and this vector, we can represent the resulting state of
process $p_i$ by the pair $(i,\tau)$.

Given $\delta(v)$ for each $v \in \sigma$, we can define the graph $\chi(\sigma,\delta)$, representing the configurations reachable from $C$
by 1-round schedules, as follows:
\begin{itemize}
\item
$v$ is a vertex in $\chi(\sigma,\delta)$ if and only if
\begin{itemize}
\item
$v$ is a terminated vertex in $\sigma$
or
\item
$v = (i,\tau)$, where $\tau$ is a subset of the active vertices in $\sigma$ and $i \in \mathit{id}(\tau)$.
\end{itemize}
If $v = (i,\tau)$, then $\mathit{id}(v) = i$.
\item
$\{v,v'\}$ is an edge in $\chi(\sigma,\delta)$ if and only if $\mathit{id}(v) \neq \mathit{id}(v')$ and 
\begin{itemize}
\item
at least one of $v$ and $v'$ is a terminated vertex in $\sigma$ or
\item
$v = (\mathit{id}(v), \tau)$ and $v' = (\mathit{id}(v'), \tau')$,
where $\tau$ and $\tau'$ are subsets of the active vertices in $\sigma$ such that
$\tau \subseteq \tau'$ or $\tau' \subseteq \tau$.
\end{itemize}
\end{itemize}

A vertex is in both $\sigma$ and $\chi(\sigma,\delta)$ if and only if it is terminated.
If vertex $(i,\tau)$ is  in $\chi(\sigma,\delta)$, but not in $\sigma$, then it represents the state of process $p_i$ 
immediately after it has performed its \id{scan} of $S_{t+1}$, $\tau$ represents the result of the \id{scan}, 
and $\mathit{id}(\tau)$ is the set of identifiers of the processes that performed an  \id{update} to $S_{t+1}$ prior to this \id{scan}.
Note that $i \in \mathit{id}(\tau)$, since process $p_i$ performs its \id{update} to $S_{t+1}$ before its \id{scan}.

\figureref{fig:subdivision} illustrates the subdivisions of two different 3-vertex cliques. In $\sigma$, all three vertices are active,
the state of process $p_1$ is represented by vertex $x$, the state of $p_2$ is represented by vertex $y$, and the state of $p_3$ is represented by vertex $z$.
In $\sigma'$, process $p_1$ and $p_3$ have the same states, but the state of $p_2$ is represented by vertex $y'$,
which is terminated. 
For readability, process identifiers are omitted from the representation of states in $\chi(\sigma,\delta)$ and $\chi(\sigma',\delta)$.
Instead, white vertices indicate states of $p_1$, red vertices indicate states of $p_2$, and black vertices indicate states of $p_3$.
  
\begin{figure*} 
\centering
\includegraphics[width=0.7\textwidth, scale=0.5]{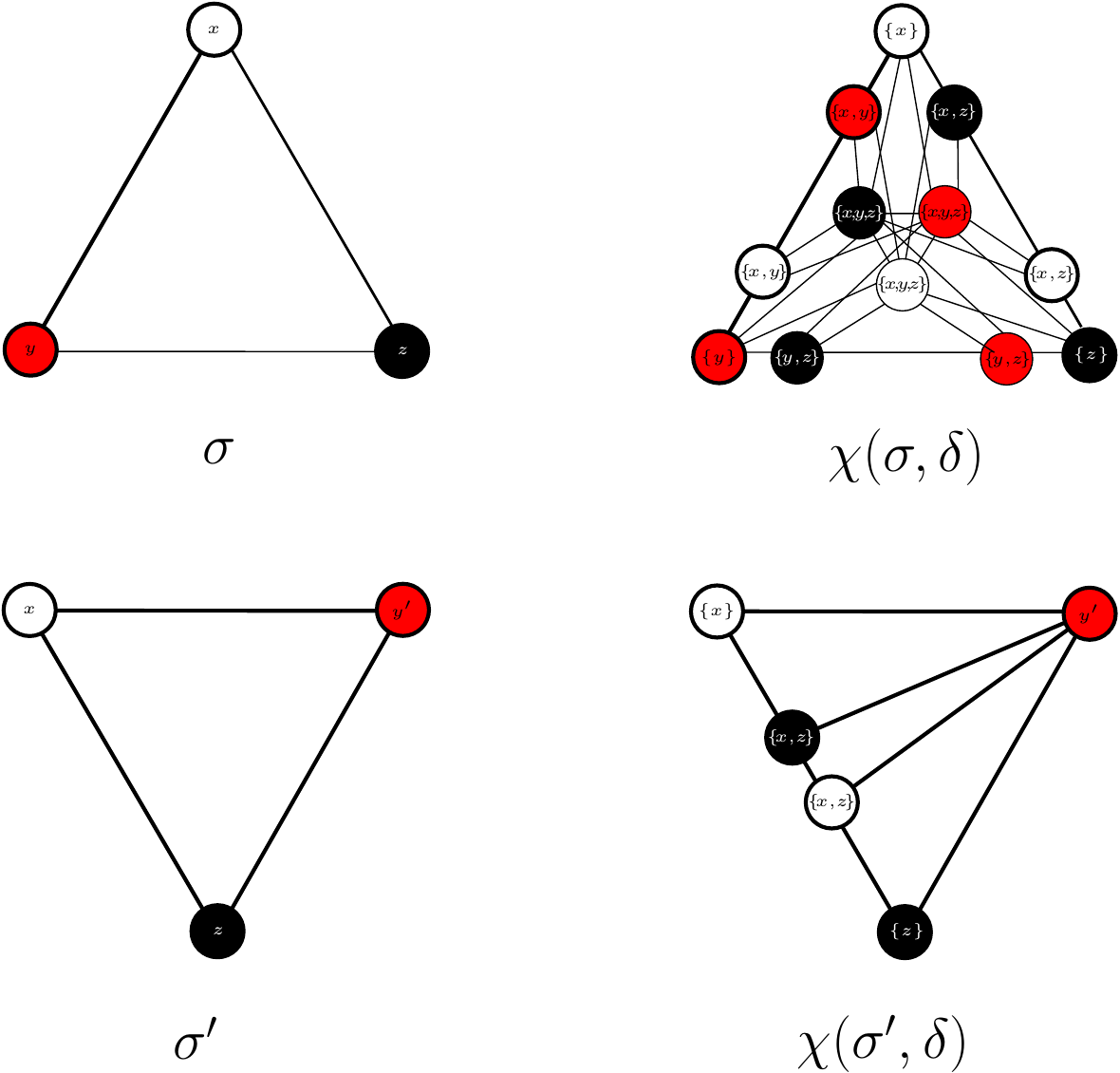}
\caption{The subdivisions of two 3-vertex cliques.}
\label{fig:subdivision}
\end{figure*}

The next two results show that there is a correspondence between the $n$-vertex cliques in $\chi(\sigma,\delta)$
and the configurations reachable from $C$ by 1-round schedules.

\begin{lemma}
Let $\sigma$ be an $n$-vertex clique
that represents a configuration $C$ in which all active processes are poised to \id{update} the same snapshot object.
If $\beta$ is a 1-round schedule from $C$, then
the configuration $C\beta$ is represented by an $n$-vertex clique in $\chi(\sigma,\delta)$.
\label{lem:config-clique}
\end{lemma}

\begin{proof}
Let $A$ be the set of active processes in configuration $C$
and let $S_{t+1}$ be the snapshot object the processes in $A$ are poised to \id{update}.
Then a 1-round schedule $\beta$ from $C$ is a sequence consisting of two copies of each process in $A$.
Consider the second occurrence in $\beta$ of a process $p_i$. This corresponds to the step in the schedule $\beta$
at which $p_i$ performs its \id{scan} of $S_{t+1}$. Let $\alpha$ be the prefix of $\beta$ prior to this step.
For each $p_j \in A$, $p_j$ occurs in $\alpha$ if and only if $S_{t+1}[j]$ in configuration $C\alpha$ contains the state of $p_j$
in configuration $C$.
Let $\tau$ be the set of vertices in $\sigma$ that represent the states of processes appearing in $S_{t+1}$ in configuration $C\alpha$.
Then $\tau$ represents the result of the  \id{scan} of $S_{t+1}$ by process $p_i$.
In particular, $p_i$ occurs in $\alpha$,  since it performs its \id{update} to $S_{t+1}$ before its \id{scan}.
Hence $i \in \mathit{id}(\tau)$ and $(i,\tau) \in \chi(\sigma,\delta)$.

Suppose $j \neq i$, $p_j \in A$, and
the second occurrence of $p_j$ in $\beta$ occurs after $p_i$.
Let $\rho$ be the subset of $\sigma$ representing the result of $p_j$'s \id{scan}, so $(j,\rho) \in \chi(\sigma,\delta)$.
Then the prefix of $\beta$ prior to the second occurrence of $p_j$ begins with $\alpha$.
Hence  $\tau \subseteq \rho$ and $\{(i,\tau), (j,\rho)\}$ is an edge in $\chi(\sigma,\delta)$.

For each process $p_i$ that is terminated in $C$, the vertex in $\sigma$ with identifier $i$ is also in $\chi(\sigma,\delta)$
and this vertex is connected to every other vertex in $\chi(\sigma,\delta)$ with a different identifier.
Thus the configuration $C\beta$ is represented by an $n$-vertex clique in $\chi(\sigma,\delta)$.
\end{proof}

\begin{lemma}
Let $\sigma$ be an $n$-vertex clique
that represents a configuration $C$  in which all active processes are poised to \id{update} the same snapshot object.
Every $n$-vertex clique in $\chi(\sigma,\delta)$ represents a configuration reachable from
$C$ by a 1-round schedule.
\label{lem:clique-config}
\end{lemma}

\begin{proof}
Let $S_{t+1}$ be the snapshot object the active processes in $C$ are poised to \id{update}.
Let $\sigma'$ be an $n$-vertex clique in $\chi(\sigma,\delta)$.
Since $\mathit{id}(v) \neq \mathit{id}(v')$ for all edges $\{v,v'\}$ in $\chi(\sigma,\delta)$,  
there is vertex $v_i \in \sigma'$ with $\mathit{id}(v_i) = i$ for each $i \in \{1,\ldots, n\}$.
Let $I$ be the set of indices of active processes in configuration $C$. 
The definition of $\chi(\sigma,\delta)$ implies that $v_i \in \sigma$, for each $i \not\in I$, and
$v_i = (i,\tau_i)$, for each $i  \in I$, where $\tau_i$ is a subset of the active vertices in $\sigma$ and $i \in \mathit{id}(\tau_i)$.
Furthermore, if $i,j \in I$ and $i \neq j$, then either $\tau_i \subseteq \tau_j$ or $\tau_j \subseteq \tau_i$,
since $\{v_i,v_j\}$ is an edge of $\chi(\sigma,\delta)$.
Since $\tau_i \subseteq \tau_j$ implies $\mathit{id}(\tau_i) \subseteq \mathit{id}(\tau_j)$, the sets
$\mathit{id}(\tau_i)$  for $i \in I$ can be ordered by inclusion.
Let $\prec$ be a total order on $I$ such that if $i$ occurs in more of these sets than $j$ does, then $i \prec j$.
In other words, for all $i \in I$, the elements of $\mathit{id}(\tau_i)$ occur before the elements of $I-\mathit{id}(\tau_i)$.

Let $\alpha'$ be a sequence containing one copy of each process whose identifier is in $I$ such that $p_i$ occurs before $p_j$ if and only if $i \prec j$.
If the schedule $\alpha'$ is performed starting from $C$, then, for each $i \in I$, $\tau_i$ represents the contents of $S_{t+1}$  at some point during the execution.
Note that, since $i \in \mathit{id}(\tau)$, this point occurs after $p_i$ performs its \id{update}.
For each $i \in I$, insert a second copy of $p_i$ after the process in $\alpha'$ whose \id{update}
causes the contents of $S_{t+1}$ to be represented by $\tau_i$ and before the next process in $\alpha'$.
Let $\alpha$ be the resulting sequence.
Then $\alpha$ is a 1-round schedule such that $v_i = (i,\tau_i)$ represents the state of $p_i$ in  configuration $C\alpha$, for each $i \in I$.
For each $i \not\in I$, $v_i \in \sigma$,
so it represents the state of the terminated process $p_i$ in $C$ and, thus, the state of $p_i$ in $C\alpha$, too.
Hence $\sigma'$ represents the configuration $C\alpha$.
\end{proof}

For any two graphs $G=(V,E)$ and $G' = (V', E')$,
the \emph{union} of  $G$ and $G'$ is the graph $G \cup G' = (V\cup V', E \cup E')$.
Then $\gt{t}$ is a union of $n$-vertex cliques.
Consider any subgraph $\at$ of $\gt{t}$ that is the union of $n$-vertex cliques.
We define  $\chi(\at,\delta)$ to be the union of the graphs  $\chi(\sigma,\delta)$ for all $n$-vertex cliques $\sigma$ in $\at$.
In particular,  $\chi(\gt{t},\delta)$ is the union of $\chi(\sigma,\delta)$ for all $n$-vertex cliques $\sigma$ in $\gt{t}$.
By \lemmaref{lem:config-clique} and  \lemmaref{lem:clique-config}, it follows that $\gt{t+1} = \chi(\gt{t},\delta)$.
This method for obtaining 
$\gt{t+1}$ from $\gt{t}$ is closely related to
the non-uniform chromatic subdivision of
a simplicial complex representing a protocol in the NIIS model, introduced by Hoest and Shavit~\cite{HS2006}.
Consequently, we will call the graph $\chi(\sigma,\delta)$ the \emph{subdivision} of the clique $\sigma$
and the graph $\gt{t+1}$ the \emph{subdivision} of the graph $\gt{t}$.
However, Herlihy and Shavit~\cite[page 884]{HS99} mention that $\chi(\sigma,\delta)$ is not necessarily a topological subdivision of $\sigma$.

By definition, a clique is connected. We show that subdivision of a clique in $\gt{t}$ is still connected.

\begin{lemma}
\label{lem:simp-conn}
The  subdivision $\chi(\sigma,\delta)$ of every $n$-vertex clique $\sigma$ in $\gt{t}$ is connected.
\end{lemma}
\begin{proof}
First suppose that $\sigma$ contains some terminated vertex $u$. By definition, $u \in \chi(\sigma,\delta)$
and no other vertex of $\chi(\sigma,\delta)$ has the same $id$.
Consider any other vertex $v \in \chi(\sigma,\delta)$. By definition, $\{u,v\}$ is an edge in $\chi(\sigma,\delta)$. 
Thus every vertex in $\chi(\sigma,\delta)$ is connected to $u$, so the graph $\chi(\sigma,\delta)$ is connected.

Now suppose that $\sigma$ contains only active vertices. Then, for each $i \in \{1,\ldots,n\} = \mathit{id}({\sigma})$,
$(i, \sigma) \in \chi(\sigma,\delta)$. Furthermore, if $i,j \in \{1,\ldots,n\}$ and $i \neq j$, then
$\{(i,\sigma),(j,\sigma)\}$ is an edge in  $\chi(\sigma,\delta)$, so the
vertices $(i,\sigma)$ for $i \in \{1,\ldots,n\}$ form an $n$-vertex clique.
Now consider any vertex $(i, \tau) \in \chi(\sigma,\delta)$, where $\tau \subsetneq \sigma$.
Then, by definition, $\{(i,\tau), (j,\sigma)\}$ is an edge in $\chi(\sigma,\delta)$
for any $j \in \{1,\ldots,n\} -\{i\}$. Since $n \geq 2$, such a $j$ exists.
Thus every vertex in $\chi(\sigma,\delta)$ is connected to this clique, so the graph $\chi(\sigma,\delta)$ is connected.
\end{proof}

More generally, connectivity is preserved by subdivision.

\begin{lemma}
\label{lem:nucs-connected}
Let $\at$ be a connected subgraph of $\gt{t}$ that is the union of $n$-vertex cliques.
Then $\chi(\at,\delta)$ is a connected subgraph of $\gt{t+1}$.
\end{lemma}

\begin{proof}
Consider any two vertices $u',v' \in \chi(\at,\delta)$. Then $u' \in \chi(\sigma,\delta)$ and $v' \in \chi(\tau,\delta)$
for some $n$-vertex cliques $\sigma,\tau \subseteq \at$.
Since $\at \subseteq \gt{t}$ is connected, there is a path $w_0,\ldots ,w_\ell$ in $\at$
of length $\ell \geq 0$ in $\at$ such that $w_0 \in \sigma$ and $w_\ell \in \tau$.
For $0 \leq i \leq \ell$, let $w'_i = w_i$ if $\delta(w_i) \neq \bot$,  and let $w'_i = (\mathit{id}(w_i),\{w_i\})$
 if $\delta(w_i) = \bot$.

Consider any $i$ such that $1 \leq i \leq \ell$. Since $\{w_{i-1},w_i\}$ is an edge of $\at$, there exists an $n$-vertex clique $\sigma_i \subseteq \at$ that contains this edge. Since $w_{i-1},w_i \in \sigma_i$, it follows by construction that $w'_{i-1},w'_i \in \chi(\sigma_i,\delta)$.
By \lemmaref{lem:simp-conn}, the subdivision
$\chi(\sigma_i,\delta)$ of $\sigma_i$ is connected. Thus, there exists a path between $w'_{i-1}$ and $w'_i $ in 
$\chi(\sigma_i,\delta) \subseteq \chi(\at,\delta)$.
By \lemmaref{lem:simp-conn},
$\chi(\sigma,\delta)$ and $\chi(\tau,\delta)$
are connected, so there exist a path between
$u'$ and $w'_0$ in $\chi(\sigma,\delta) \subseteq \chi(\mathbb{A},\delta)$ and 
a path between
$w'_\ell$ and $v'$ in $\chi(\tau,\delta) \subseteq \chi(\mathbb{A},\delta)$.
Hence, there is a path between $u'$ and $v'$
in $\chi(\mathbb{A},\delta)$.

Since $u'$ and $v'$ are arbitrary, $\chi(\at,\delta)$ is connected.
\end{proof}

The next result follows by induction, because $\gt{t+1} = \chi(\gt{t},\delta)$.

\begin{corollary}
    \label{cor:nucs-connected}
    If $\gt{0}$ is connected, then, for all $t \geq 1$, $\gt{t}$ is connected.
\end{corollary}

If $\Tt \subseteq \vt{t}$ is a set of terminated vertices in $\gt{t}$,
we define $\chi(\Tt,\delta) = \Tt \subseteq \vt{t+1}$.
Let $\at$ and $\bt$ each be either a nonempty set of terminated vertices in $\gt{t}$ or the nonempty union of $n$-vertex cliques in $\gt{t}$.
Then the \emph{distance between} $\at$ and $\bt$ (in $\gt{t}$) is the minimum of the length of the paths
between $u \in \at$ and $v \in \bt$.
If $\gt{0}$, is connected, then \corollaryref{cor:nucs-connected} implies that 
at least one such path exists.
Now we show that if the distance between $\at$ and $\bt$ is 0 (i.e. they intersect),
  then the same is true for $\chi(\mathbb{A},\delta)$ and $\chi(\mathbb{B},\delta)$ 
  and, if the distance between $\at$ and $\bt$ is greater than  0 (i.e., they are disjoint), then so are
$\chi(\mathbb{A},\delta)$ and $\chi(\mathbb{B},\delta)$.

\begin{lemma}
    \label{lem:nucs-disjoint}
Suppose $\at$ and $\bt$ are each either a set of terminated vertices in $\gt{t}$ or the union of $n$-vertex cliques in $\gt{t}$.
Then  $\at$ and $\bt$ are disjoint if and only if
$\chi(\at,\delta)$ and $\chi(\bt,\delta)$ are disjoint. 
\end{lemma}

\begin{proof}
When $\at$ or $\bt$ is a set of terminated vertices in $\gt{t}$, any vertex $u \in  \at \cap \bt$ is terminated, so
$u \in \at \cap \bt$ if and only if $u \in \chi(\at,\delta) \cap \chi(\bt,\delta)$.
So, assume that $\at$ and $\bt$ are the unions of $n$-vertex cliques. 

Suppose that $\at$ and $\bt$ share a common vertex $u$.
Let $\sigma$ be an $n$-vertex clique in $\at$ that contains $u$ and 
let $\rho$ be an $n$-vertex clique in $\bt$  that contains $u$.
If $u$ is a terminated vertex in $\gt{t}$, then, by definition, 
$u$ is a vertex in both $\chi(\sigma,\delta)$ and $\chi(\rho,\delta)$.
Otherwise, $u$ is active in $\gt{t}$. In this case, let $\tau = \{u\}$ and $i = \mathit{id}(u)$.
Then $i \in {\mathit id}(\tau)$,  $\tau \subseteq \sigma$, and $\tau \subseteq \rho$.
By definition $(i,\tau)$ is a vertex in both  $\chi(\sigma,\delta)$ and $\chi(\rho,\delta)$.
Since $\chi(\sigma,\delta)$ is a subgraph of $\chi(\mathbb{A},\delta)$
and $\chi(\rho,\delta)$ is a subgraph of $\chi(\mathbb{B},\delta)$, in both cases it follows that
$\chi(\mathbb{A},\delta)$ and $\chi(\mathbb{B},\delta)$ are not disjoint.

Conversely, suppose that $\chi(\mathbb{A},\delta)$ and $\chi(\mathbb{B},\delta)$ share a common vertex $v$.
By definition, there exists an $n$-vertex clique 
$\sigma \subseteq \at$,  such that $v \in \chi(\sigma,\delta)$.
Similarly, there exists an $n$-vertex clique 
$\rho \subseteq \bt$ such that $v \in \chi(\rho,\delta)$.
If $v$ is a terminated vertex in $\gt{t}$, then
$v$ is a vertex in both $\sigma$ and $\rho$.
Otherwise, $v= (i,\tau)$ where $i \in {\mathit id}(\tau)$, 
$\tau \subseteq \sigma$, and $\tau \subseteq \rho$. Hence, in both cases,  $\at$ and $\bt$ are not disjoint.
\end{proof}

\lemmaref{lem:nucs-disjoint} can be generalized to show that subdividing does not decrease distances.

\begin{lemma}
	\label{lem:dist1}
Suppose $\at,\bt\subseteq \gt{t}$ are nonempty and each is
either a set of terminated vertices or the union of $n$-vertex cliques.
Then the distance between $\chi(\at,\delta)$ and $\chi(\bt,\delta)$ in $\gt{t+1}$
is at least as large as the distance between $\at$ and $\bt$ in $\gt{t}$.
\end{lemma}

\begin{proof}
Let $d$ be the distance between $\at$ and $\bt$ in $\gt{t}$. The proof is by induction on $d$.
If $d =0$, then the claim is true, since distances are always non-negative.
If $d =1$, then $\at$ and $\bt$ are disjoint. By \lemmaref{lem:nucs-disjoint}, $\chi(\at,\delta)$ and $\chi(\bt,\delta)$
are also disjoint, so the distance between them is at least 1.

Now suppose that $d \geq 2$ and the claim is true for all nonempty 
$\at^*,\bt^* \subseteq \gt{t}$
such that the distance between $\at^*$ and $\bt^*$ is $d-1$
and each is either a set of terminated vertices or the union of $n$-vertex cliques. 
Consider any vertex $v \in \vt{t}$ at distance 1 from $\at$. Then there exists
a vertex $u\in \at$ such that $\{u,v\} \in \et{t}$.
Since $\gt{t}$ is a union of $n$-vertex cliques, there exist $n-2$ other vertices that form a clique with $\{u,v\}$.
Since these vertices are adjacent to $u$, they are all at distance at most 1 from $\at$.
Let $\at'$ denote the union of all $n$-vertex cliques in $\gt{t}$ that contain at least one vertex in $\at$.

Consider any path $u_0, u_1, \ldots, u_d$ of length $d$ between $\at$ and $\bt$ in $\gt{t}$. 
Note that $u_1 \not\in \at$, since the distance between $\at$ and $\bt$ in $\gt{t}$ is $d$.
Thus $u_1$ is a vertex at distance 1 from $\at$ and, hence, is in $\at'$. Therefore the distance between $\at'$ and $\bt$ in $\gt{t}$ is at most $d-1$.
In fact, the distance between $\at'$ and $\bt$ is exactly $d-1$.
Suppose not.
Then there exists a path $v_0, \ldots, v_\ell$ in $\gt{t}$ between $\at'$ and $\bt$ where $\ell < d-1$.
If $v_0 \in \at$, then this path is between $\at$ and $\bt$. 
If $v_0 \in \at' - \at$, then, by
definition of $\at'$, there exists a vertex $u \in \at$ such that $\{u,v_0\} \in \et{t}$.
But then $u,v_0,\ldots,\ell$ is a path between  $\at$ and $\bt$.
In both cases, this shows that the distance between $\at$ and $\bt$ is less than $d$, which 
contradicts the definition of $d$.

Consider any shortest path $w_0,w_1,\ldots,w_\ell$ between $\chi(\at,\delta)$ and $\chi(\bt,\delta)$ in $\gt{t+1}$.
Note that $w_1 \not\in \chi(\at,\delta)$, since this is a shortest path. By definition of $\gt{t+1}$, $\{w_0,w_1\}$ is an edge of
$\chi(\sigma,\delta)$ for some $n$-vertex clique $\sigma \subseteq \gt{t}$.
If $w_0$ is terminated in $\gt{t}$, then $w_0 \in \at$ and $w_0 \in \sigma$,
so, by definition, $\sigma$ is in $\at'$.
Otherwise, 
since $w_0 \in \chi(\sigma,\delta)$,
there is a process identifier $i$ and a set of vertices $\tau_i \subseteq \sigma$
such that $w_0 = (i,\tau_i)$  and $i \in {\mathit id}(\tau_i)$.
Moreover, since $w_0 \in \chi(\at,\delta)$,
there exists an $n$-vertex clique $\rho \subseteq \at$
such that $w_0 \in \chi(\rho,\delta)$, so $\tau_i \subseteq \rho$.
In this case, let  
$v_i \in \tau_i$ be such that ${\mathit id}(v_i) = i$.
Since $\tau_i \subseteq \rho \subseteq \at$, we have $v_i \in \at$ and, since $\tau_i \subseteq \sigma$, we have $\sigma \in \at'$.
Hence, in both cases, $w_1 \in \chi(\at',\delta)$.
By the induction hypothesis, the distance between $\chi(\at',\delta)$ and $\chi(\bt,\delta)$ in $\gt{t+1}$ is at least $d-1$.
Thus, $\ell \geq d$.
\end{proof}

If $\at$ and $\bt$ are disjoint unions of $n$-vertex cliques
and $\ct$ is a union of $n$-vertex cliques all of whose vertices are active,     
then there is no edge in the subdivision of $\ct$ that connects the subdivisions of $\at$ and $\bt$.

\begin{lemma}
    \label{lem:nucs-distincrease}
Suppose $\at$, $\bt$, and $\ct$ are nonempty unions of $n$-vertex cliques in $\gt{t}$,
$\at \cap \ct$ is nonempty,
$\bt \cap \ct$ is nonempty,
$\at$ and $\bt$ are disjoint,
and all vertices in $\ct$ are active.
Then the distance between $\chi(\at,\delta)\cap \chi(\ct,\delta)$ and $\chi(\bt,\delta)\cap \chi(\ct,\delta)$ in $\gt{t+1}$ is at least 2.
\end{lemma}

\begin{proof}
Since $\at \cap \ct$ and $\bt \cap \ct$ are nonempty, \lemmaref{lem:nucs-disjoint} says that $\chi(\at,\delta) \cap \chi(\ct,\delta)$ and $\chi(\bt,\delta) \cap \chi(\ct,\delta)$ are nonempty.
Since $\at$ and $\bt$ are disjoint, 
it also
says that the distance between 
$\chi(\at,\delta)$ and $\chi(\bt,\delta)$ in $\gt{t+1}$ is at least 1.
Hence, the distance between $\chi(\at,\delta)\cap \chi(\ct,\delta)$ and $\chi(\bt,\delta)\cap \chi(\ct,\delta)$ in $\gt{t+1}$ is at least 1.

To obtain a contradiction, suppose that the distance between $\chi(\at,\delta)\cap \chi(\ct,\delta)$ and $\chi(\bt,\delta)\cap \chi(\ct,\delta)$ in $\gt{t+1}$ is 1.
Then there exist vertices  $u \in\chi(\at,\delta)\cap \chi(\ct,\delta)$ and $v \in\chi(\bt,\delta)\cap \chi(\ct,\delta)$
such that $\{u,v\} \in \et{t+1}$. 
Since $u,v \in \chi(\ct,\delta)$ and all vertices in $\ct$ are active, 
$u = (i, \tau_i)$ and $v = (j, \tau_j)$,
where $i \in \mathit{id}(\tau_i)$, $j \in \mathit{id}(\tau_j)$, $\tau_i \subseteq \sigma_i$, and $\tau_j \subseteq \sigma_j$ 
for  some $n$-vertex cliques $\sigma_i,\sigma_j \subseteq \ct$.
Since $u \in \chi(\at,\delta)$, it follows that $\tau_i \subseteq \rho_i$ for some $n$-vertex clique $\rho_i \subseteq \at$.
Similarly,  $\tau_j \subseteq \rho_j$ for some $n$-vertex clique $\rho_j\subseteq \bt$.
Since $\{u,v\} \in \et{t+1}$, $i \neq j$ and either
$\tau_i \subseteq \tau_j \subseteq \bt$ or $\tau_j \subseteq \tau_i \subseteq \at$.
In both cases, $\at$ and $\bt$ are not disjoint, contrary to assumption.
\end{proof}	

We now prove one of the main technical tools used in this paper. It shows that the distance between $\at$ and $\bt$ in $\gt{t}$ is less than the distance between their subdivisions in $\gt{t+1}$, provided that there is no path between $\at$ and $\bt$ 
in which every edge 
contains at least one terminated vertex.

\figureref{fig:dist} illustrates \lemmaref{lem:dist2}. In the top diagram, which is part of $\gt{t}$, the grey triangle represents $\mathbb{A}$, which consists of one 3-vertex clique and $\mathbb{B}= \{v_4\}$ is a set containing one terminated vertex.
The blue path, which has length 4, is a shortest path between $\mathbb{A}$ and $\mathbb{B}$ in $\gt{t}$.
Note that $v_2$ and $v_3$ are both active vertices.
In the bottom diagram, which is part of $\gt{t+1}$, the grey triangle represents $\chi(\mathbb{A},\delta)$
and $\chi(\mathbb{B},\delta) = \mathbb{B}$.
The blue path, which now has length 5,  is a shortest path between $\chi(\mathbb{A},\delta)$ and $\chi(\mathbb{B},\delta)$ in $\gt{t+1}$.
  
\begin{figure*} 
\centering
\includegraphics[width=0.7\textwidth, scale=0.5]{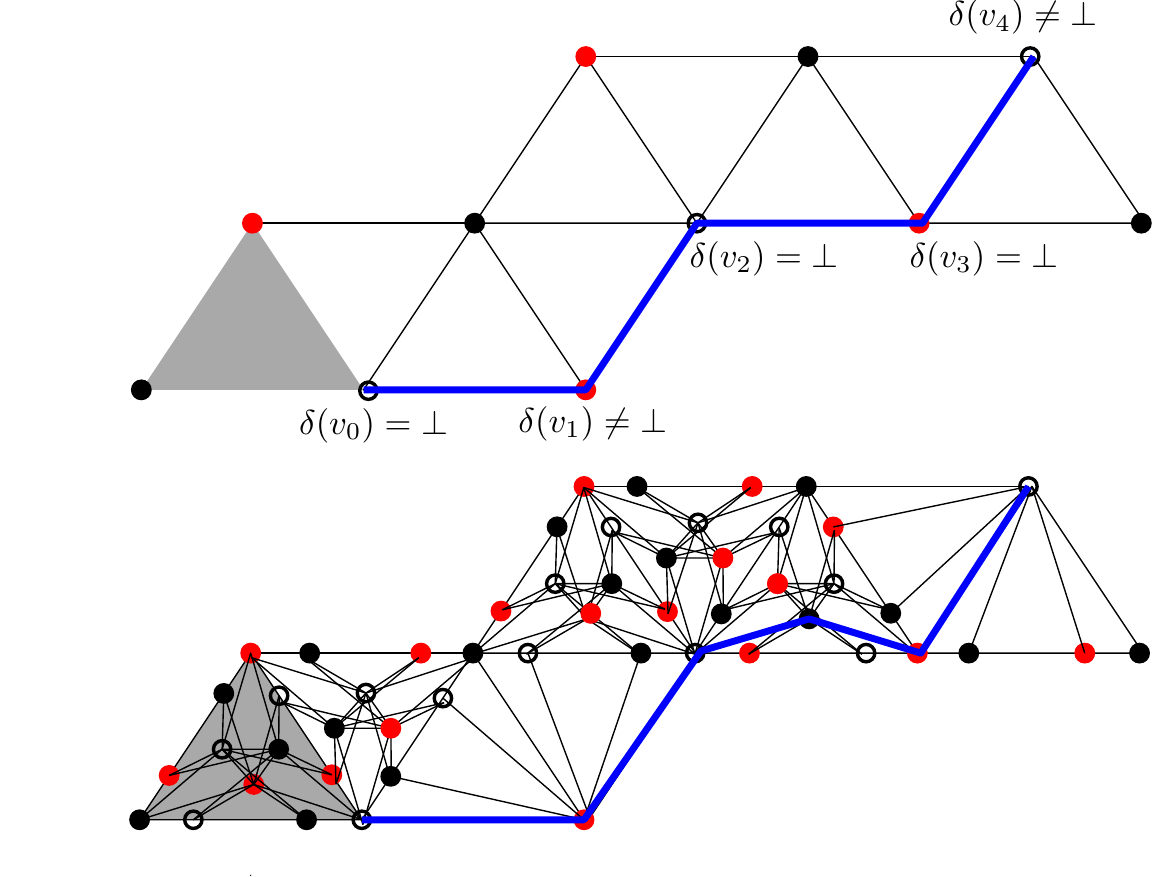}
\caption{An illustration of \lemmaref{lem:dist2}.}
\label{fig:dist}
\end{figure*}
	
\begin{lemma}
	\label{lem:dist2}
Suppose $\at,\bt\subseteq \gt{t}$ are nonempty and each is either a set of terminated vertices or the union of $n$-vertex cliques.
If every path between $\mathbb{A}$ and $\mathbb{B}$ in $\gt{t}$ contains at least one edge between active vertices, then 
the distance between $\chi(\at,\delta)$ and $\chi(\bt,\delta)$ in $\gt{t+1}$ is larger than the distance between $\at$ and $\bt$ in $\gt{t}$.
\end{lemma}

\begin{proof}
Assume that every path between $\at$ and $\bt$ in $\gt{t}$ contains at least one edge between active vertices.
Consider any such path $v_0,\ldots,v_\ell$ between $\at$ and $\bt$.
Suppose that, for all $1 \leq i \leq \ell$, the edge $\{v_{i-1},v_i\}$ is contained in an $n$-vertex clique that contains a terminated vertex $u_i$.
Replace each edge $\{v_{i-1},v_i\}$ that is  between active vertices by the subpath $v_{i-1},u_i,v_i$.
The result is a path between $\at$ and $\bt$ in $\gt{t}$ that contains no edges
between active vertices, contrary to our assumption.
Therefore, every path between $\at$ and $\bt$ in $\gt{t}$ contains at least one edge such that every $n$-vertex clique which contains this edge
is comprised of
active vertices.

Let $\ct$ be the union of a minimal set of $n$-vertex cliques in $\gt{t}$
comprised of
active vertices
such that every path between
$\at$ and $\bt$ in $\gt{t}$ contains at least one edge in $\ct$.
Let $\ft$ be the union 
of all other cliques in $\gt{t}$. Then there are no paths between $\at$ and $\bt$ in $\ft$.
Let $\at'$ be the union of all cliques in $\ft$ that are connected to $\at$
and let $\bt'$ be the union of all 
cliques in $\ft$ that are connected to $\bt$.
If $\at$ is the union of $n$-vertex cliques, then $\at'$ is nonempty, since $\at \subseteq \at'$.
If $\at$ is a set of terminated vertices, consider the first vertex of some path from $\at$ to $\bt$ in $\gt{t}$.
By definition, it is contained in some $n$-vertex clique $\sigma \subseteq \gt{t}$.
Since the first vertex of this path is in $\at$ and all vertices in $\ct$ are active, $\sigma \not\subseteq \ct$.
Hence $\sigma \subseteq \at'$, so $\at'$ is nonempty.
Similarly, 
$\bt'$ is nonempty.

Every path between $\at'$ and $\bt'$ in $\gt{t}$ contains at least one edge in $\ct$.
Consider any path $v_0,\ldots,v_\ell$ between $\at'$ and $\bt'$. Suppose $\{ v_i, v_{i+1} \}$ is the first edge on this path
that is contained in $\ct$. Then $v_i \in \at'$. Hence every path between $\at'$ and $\bt'$ in $\gt{t}$ and, hence, 
every path between $\at$ and $\bt$ in $\gt{t}$ contains an edge in $\ct$ 
with one endpoint in $\at'$.
By the minimality of $\ct$, every clique in $\ct$ intersects $\at'$.
Similarly, every clique in $\ct$  
intersects $\bt'$.
Therefore,
every shortest path between $\at \subseteq \at'$ and $\bt \subseteq \bt'$ in $\gt{t}$ 
consists of a path between $\at$ and $\at' \cap \ct$, followed by an edge between $\at' \cap \ct$ and $\bt' \cap \ct$,
followed by a path between $\bt' \cap \ct$ and $\bt$.

Since
$\gt{t+1} = \chi(\gt{t},\delta)$ is the union of $\chi(\sigma,\delta)$ for all $n$-vertex cliques $\sigma$ in $\gt{t}$,
it follows that $\gt{t+1}$ is the union of the  $n$-vertex cliques in $\chi(\at',\delta)$, $\chi(\bt',\delta)$, and $\chi(\ct,\delta)$.
Furthermore, since $\at'$ and $\bt'$ are disjoint, \lemmaref{lem:nucs-disjoint} implies that $\chi(\at',\delta)$ and $\chi(\bt',\delta)$
are disjoint. Thus, every path between  $\chi(\at,\delta) \subseteq \chi(\at',\delta)$ and $\chi(\bt,\delta) \subseteq \chi(\bt',\delta)$ in $\gt{t+1}$ 
consists of a path between $\chi(\at,\delta)$ and $\chi(\at',\delta) \cap \chi(\ct,\delta)$, followed by a path between $\chi(\at',\delta) \cap \chi(\ct,\delta)$ and $\chi(\bt',\delta) \cap \chi(\ct,\delta)$,
followed by a path between $\chi(\bt',\delta) \cap \chi(\ct,\delta)$ and $\chi(\bt,\delta)$.

Since $\chi(\at',\delta) \cap \chi(\ct,\delta) \subseteq \chi(\ct,\delta)$, 
the distance between $\chi(\at,\delta)$ and $\chi(\at',\delta) \cap \chi(\ct,\delta)$ in $\gt{t+1}$ is at least as large as
the distance between $\chi(\at,\delta)$ and $\chi(\ct,\delta)$ in $\gt{t+1}$.
By \lemmaref{lem:dist1}, the distance between $\chi(\at,\delta)$ and $\chi(\ct,\delta)$ in $\gt{t+1}$ is at least as large 
as the distance between $\at$ and $\ct$ in $\gt{t}$.
The distance between $\at$ and $\ct$ in $\gt{t}$ is equal to the distance between $\at$ and $\at' \cap \ct$ in $\gt{t}$, because
$\at'$ is the union of all cliques in $\ft$ that are connected to $\at$.
Hence, the distance between $\chi(\at,\delta)$ and $\chi(\at',\delta) \cap \chi(\ct,\delta)$ in $\gt{t+1}$ is at least as large as
the distance between $\at$ and $\at' \cap \ct$ in $\gt{t}$.
Similarly,  the distance between $\chi(\bt',\delta) \cap \chi(\ct,\delta)$  and $\chi(\bt,\delta)$ in $\gt{t+1}$ is at least as large as the distance between $\bt' \cap \ct$ and $\bt$ in $\gt{t}$.
By \lemmaref{lem:nucs-distincrease}, the distance between $\chi(\at',\delta)\cap \chi(\ct,\delta)$ and $\chi(\bt',\delta)\cap \chi(\ct,\delta)$ in $\gt{t+1}$ is at least 2.
Therefore,

\begin{tabular}{l}
the distance between $\chi(\at,\delta)$ and $\chi(\bt,\delta)$ in $\gt{t+1}$\\
$\geq$ the distance between $\chi(\at,\delta)$ and $\chi(\at',\delta) \cap \chi(\ct,\delta)$ in $\gt{t+1}$\\
$+$ the distance between $\chi(\at',\delta) \cap \chi(\ct,\delta)$ and $\chi(\bt',\delta) \cap \chi(\ct,\delta)$ in $\gt{t+1}$\\
$+$ the distance between $\chi(\bt',\delta) \cap \chi(\ct,\delta)$ and $\chi(\bt,\delta)$ in $\gt{t+1}$\\
$\geq$ the distance between $\at$ and $\at' \cap \ct$ in $\gt{t}$\\
$+$ 2\\
$+$ the distance between $\bt' \cap \ct$ and $\bt$ in $\gt{t}$\\
$>$ the distance between $\at$ and $\bt$ in $\gt{t}$.
\end{tabular}

\end{proof}
\section{Why Extension-Based Proofs Fail}
\label{sec:ext-extensions}
\newcommand{\nta}{\mathbb{N}_t(a)}
\newcommand{\Nta}[2]{\mathbb{N}_{#1}(#2)}
\newcommand{\tta}[1]{\mathbb{T}_t(#1)}
\newcommand{\Tta}[2]{\mathbb{T}_{#1}(#2)}
\newcommand{\xta}[1]{\mathbb{X}_t(#1)}
\newcommand{\Xta}[2]{\mathbb{X}_{#1}(#2)}
\newcommand{\Rt}{\mathbb{Q}}
  
In this section, we prove that  no extension-based proof can show the
impossibility of deterministically solving $k$-set agreement in a wait-free manner in the NIS model, for $n > k \geq 2 $ processes. Specifically, we define an adversary that is able to win against every extension-based prover. The adversary maintains a \emph{partial} specification of $\delta$ (the protocol it is adaptively constructing) and an integer $t \geq 0$. The integer $t$ represents the number of times it has subdivided the input complex, $\gt{0}$.  
Once the adversary has defined $\delta$ for each vertex in $\gt{t}$,
it may subdivide $\gt{t}$, construct $\gt{t+1} = \chi(\gt{t},\delta)$, and increment $t$.

For each $0 \leq r \leq t$ and each input value $a$,
let $\Tta{r}{a}$ be the subset of terminated vertices in $\vt{r}$ that have output $a$.
The following simple property is true because
every terminated vertex remains unchanged when a subdivision is performed.

\begin{proposition}
\label{prop:tfacts}
For all input values $a$ and all $0 \leq r < t$, $\Tta{r}{a} = \chi(\Tta{r}{a},\delta)  \subseteq \Tta{r+1}{a}$.
If the adversary defines $\delta(v) = \bot$ for each vertex $v \in \gt{t}$ where $\delta(v)$ is undefined, and subdivides $\gt{t}$ to construct $\gt{t+1}$, but does not terminate any additional vertices in $\vt{t+1}$, then $\tta{a} = \Tta{t+1}{a}$.
\end{proposition}

We say that a vertex $v \in \vt{0}$ \emph{has seen} input value $a$ if it denotes the state of a process whose input has value $a$.
Inductively, we say that $v \in \vt{r+1}$ \emph{has seen} input value $a$ if $v \in \vt{r}$ and $v$ has seen $a$
or $v = (i,\tau)$ for some subset $\tau$ of active vertices of an $n$-vertex clique in $\gt{r}$ such that $i \in \mathit{id}(\tau)$
and some vertex in $\tau$ has seen $a$. In other words, if $v$ represents the state of a process $p_i$ in some configuration
reachable by an $r$-round schedule,
then $v$ has seen $a$ if and only if $p_i$ had input $x_i =a$ or,
in some round $r' \leq r$ of this schedule, there was a process $p_j$
that performed its \id{update} before $p_i$ performed its \id{scan} and the vertex representing $p_j$ has seen $a$ in round $r'-1$.
For each $r \geq 0$ and each input value $a$, let $\Nta{r}{a}$ be the the union of the $n$-vertex cliques in $\gt{r}$
none of whose vertices have seen $a$.
To avoid violating validity,  the adversary should not let any vertex in $\nta$ output the value $a$.

\begin{proposition}
\label{prop:nfacts}
For $0 \leq r < t$ and for any input value $a$,
 $\Nta{r+1}{a} = \chi(\Nta{r}{a},\delta)$.
\end{proposition}

\begin{proof}
Consider any $n$-vertex clique $\sigma \subseteq \Nta{r}{a}$.
Since no vertex in $\sigma$ has seen $a$, it follows, by definition,
that no vertex in $\chi(\sigma,\delta)$ has seen $a$.
Thus $\chi(\sigma,\delta) \subseteq \Nta{r+1}{a}$ and, hence, $\chi(\Nta{r}{a},\delta) \subseteq \Nta{r+1}{a}$.

Conversely, consider any  $n$-vertex clique $\sigma' \subseteq \Nta{r+1}{a}$.
By definition of $\gt{r+1}$, $\sigma' = \chi(\sigma,\delta)$ for some $n$-vertex clique $\sigma$ in $\gt{r}$.
If some vertex in $\sigma$ has seen $a$, then the process  in $\sigma'$ with the same $\mathit{id}$ has seen $a$.
But none of the vertices in $\sigma'$ have seen $a$, so none of the vertices in $\sigma$ have seen $a$.
Hence $\sigma \subseteq \Nta{r}{a}$ and $\sigma' \subseteq \chi(\Nta{r}{a},\delta)$.
Therefore $\Nta{r+1}{a} \subseteq \chi(\Nta{r}{a},\delta)$.
\end{proof}

In fact, every vertex in $\gt{r}$ that has not seen $a$ is in $\Nta{r}{a}$.
This is the special case of the following lemma when $\tau$ contains only one vertex.

\begin{lemma}
If $\tau$ is a subset of an $n$-vertex clique in $\gt{r}$ and no vertex in $\tau$ has seen the input value $a$,
then $\tau$ is a subset of an $n$-vertex clique in $\Nta{r}{a}$.
\end{lemma}

\begin{proof}
The proof is by induction on $r$.
Every two vertices in $\gt{0}$ that have not seen $a$ are adjacent provided they represent the states of different processes,
i.e.~they have different $\mathit{id}$s. Thus, if $\tau$ is a subset of an $n$-vertex clique in $\gt{0}$, no vertex in $\tau$ has seen the input value $a$, and $b \neq a$, then $\tau \cup \{(j,b)\ |\ j \not\in \mathit{id}(\tau)\}$ are the vertices of an $n$-vertex clique in $\Nta{0}{a}$.

Let $r \geq 0$ and assume the claim is true for $r$.
Consider any $n$-vertex clique $\sigma'$ in $\gt{r+1}$. 
Let $\tau'$ be the subset of all vertices of $\sigma'$ that have not seen $a$. Since $\gt{r+1} = \chi(\gt{r},\delta)$, there exists an $n$-vertex clique $\sigma$ in $\gt{r}$ such that $\sigma'$ is in $\chi(\sigma,\delta)$.
Let $\tau = \{v \in \sigma\ |\ \mathit{id}(v) \in \mathit{id}(\tau')\}$.
Note that, by definition, if $u \in \sigma$ has seen $a$, then every vertex $u' \in \chi(\sigma,\delta)$
with $\mathit{id}(u') = \mathit{id}(u)$ has seen $a$.
Hence, no vertex in $\tau$ has seen $a$.
By the induction hypothesis, there exists an $n$-vertex clique $\rho$ in $\Nta{r}{a}$ that contains $\tau$.

By definition of $\chi(\rho,\delta)$, it contains each vertex of $\tau'$.
Since $\tau' \subseteq \sigma'$,
the vertices in $\tau'$ are adjacent to one another.
Let $active(\rho)$ denote the set of active vertices in $\rho$ and let
$\rho' = \{ (j,active(\rho))\ |\ j  \in \mathit{id}(active(\rho))-\mathit{id}(\tau')\} \subseteq
\chi(\rho,\delta)$.
The vertices in $\rho'$ are adjacent to one another and to each vertex in $\tau'$.
Furthermore each terminated vertex in $\rho$ is adjacent to all the vertices in $\tau'$ and $\rho'$.
Hence, these vertices form an $n$-vertex clique in $\chi(\rho,\delta) \subseteq \Nta{r+1}{a}$.
Thus the claim is true for $r+1$.
\end{proof}

For each $0 \leq r \leq t$ and each input value $a$, let $\xta{a}$ be the subset of vertices in $\vt{r}$ that represent the states of processes in $P$ in configurations reachable reachable from $C$ by $P$-only schedules, for all output queries $(C,P,a)$ to which the adversary answered \none.
To avoid contradicting its responses,
the adversary should not let any vertex in $\xta{a}$ output the value $a$.

Throughout the first phase, the adversary ensures that the following invariants hold after 
its response to each query:
\begin{enumerate}
\item \label{inv:defined1} For each $0 \leq r < t$ and each vertex $v \in \vt{r}$, $\delta(v)$ is defined.

\item \label{inv:defined2} If $v \in \vt{t}$, then 
$\delta(v) \neq \bot$.

\item \label{inv:defined3}
Suppose $s$ is the state of a process in  configuration $C \in \config{A}(1) \cup \config{A}'(1)$ and
$C = C_0 \beta$ for some initial configuration $C_0 \in \config{A}(1)$.
The process occurs at most $2t+1$ times in $\beta$.
If the process occurs $2r$ times in $\beta$, then
$s \in \vt{r}$ and $\delta(s)$ is defined. 

\item \label{inv:contains} For any input value $a$, if $\tta{a}$ is nonempty, then the distance between 
$\tta{a}$ and $\Nta{t}{a}$ in $\gt{t}$ is at least 2.

\item \label{inv:outputsfar} For any two input values $a \neq b$, if $\tta{a}$ and $\tta{b}$ are nonempty, then the distance between them in $\gt{t}$ is at least 3.

\item \label{inv:x}
For every input $a$ and every vertex $v \in \xta{a}$, 
either $v \in \nta$ or there exists an input $b \neq a$ such that $v$ is distance at most 1 from $\tta{b}$.

\begin{NIIS}
\item \label{inv:none}  For every output query $(C',P',y)$ for which the adversary answered \none\ and for every
simplex $\sigma'$ in $\mathbb{S}^t$ consisting of the states of the processes in $P'$ in a configuration that is reachable from $C'$ via a $P'$-only execution, either there is some $b \neq y$ and some vertex $w \in \mathbb{T}_b^t$ that is adjacent to every vertex in $\sigma'$  or $\sigma' \subseteq \mathbb{N}_y^t$.
\end{NIIS}

\end{enumerate}

\noindent
There is nothing special about the values 2 and 3. They are simply the smallest values such that the invariants can be maintained and every chain of queries is finite.
The following lemma is a consequence of the invariants. 

\begin{lemma}
\label{lem:invariants-hold1}
For any two input values $a \neq b$,  every path between $\tta{a}$ and $\nta \cup \tta{b}$ in $\gt{t}$ contains at least one edge between active vertices.
\end{lemma}

\begin{proof}
Consider any path $v_0,v_1,\dots,v_\ell$ between $\tta{a}$ and $\nta \cup \tta{b}$ in $\gt{t}$. Let $v_j$ be the last vertex in $\tta{a}$.
Since the invariants hold after each query and $v_\ell \in \tta{b} \cup \nta$, invariants \ref{inv:contains} and \ref{inv:outputsfar} imply that the distance between $v_j$ and $v_\ell$ is at least 2. Hence, $\ell \geq j+2$. Since $v_j$ is the last vertex in $\tta{a}$, $v_{j+1}, v_{j+2} \not\in \tta{a}$. Moreover, by invariant \ref{inv:outputsfar}, $v_{j+1}, v_{j+2} \not\in \tta{c}$ for any input value $c \neq a$. Hence, $\{v_{j+1},v_{j+2}\}$ is an edge between active vertices. 
\end{proof}

Essentially, a subdivision maintains the invariants, but increases the distance
increases between vertices that output different values and between vertices that output $a$ and vertices that have not seen $a$. 

\begin{lemma}
\label{lem:invariants-hold2}
Suppose all the invariants hold,  the adversary defines $\delta(v) = \bot$ for each vertex $v \in \vt{t}$ such that $\delta(v)$ is undefined, and subdivides $\gt{t}$ to construct $\gt{t+1}$.
If $\tta{a}$ is nonempty, then the distance between $\Tta{t+1}{a}$ and $\Nta{t+1}{a}$ in $\gt{t+1}$ is greater than the distance between $\tta{a}$ and $\nta$ in $\gt{t}$.
If  $b \neq a$ and $\tta{b}$ is also nonempty, then the distance between 
$\Tta{t+1}{a}$ and $\Tta{t+1}{b}$ in $\gt{t+1}$ is greater than the distance between $\tta{a}$ and $\tta{b}$ in $\gt{t}$.
Furthermore, if the adversary increments $t$, then all the invariants hold.
\end{lemma}

\begin{proof}
By~\propositionref{prop:tfacts}, $\Tta{t+1}{a} = \tta{a} = \chi(\tta{a},\delta)$, for each input $a$. 
In addition, $\Nta{t+1}{a} = \chi(\nta,\Delta)$ by~\propositionref{prop:nfacts}.
Lemma~\ref{lem:invariants-hold1} says that every path between $\tta{a}$ and $\nta \cup \tta{b}$ in $\gt{t}$ contains at least one edge between active vertices.
Therefore, if $\tta{a}$ is nonempty,  \lemmaref{lem:dist2} implies that the distance between  $\chi(\tta{a},\delta) = \Tta{t+1}{a}$
and $\chi(\nta,\delta) = \Nta{t+1}{a}$  in $\gt{t+1}$ is greater than the distance between  $\tta{a}$ and $\nta$ in $\gt{t}$.
Similarly, if both $\tta{a}$ and $\tta{b}$ are nonempty,  \lemmaref{lem:dist2} implies that 
the distance between  $\chi(\tta{t},\delta) = \Tta{t+1}{a}$
and $\chi(\tta{b},\delta) = \Tta{t+1}{b}$  in $\gt{t+1}$ is greater than the distance between  $\tta{a}$ and $\tta{b}$ in $\gt{t}$.
Hence invariants \ref{inv:contains} and \ref{inv:outputsfar} remain true after $t$ is incremented.

Before incrementing $t$, the adversary defines $\delta(v) = \bot$ for each vertex $v \in \gt{t}$ where $\delta(v)$ was undefined.
Since invariant \ref{inv:defined1} was true, it remains true.
Invariant \ref{inv:defined2} remains true by construction and the definition of $\chi$.
Invariant \ref{inv:defined3} is not affected.
Invariant \ref{inv:x} remains true by the definition of $\chi$.
\end{proof}

\paragraph{The adversarial strategy for phase 1.} 
Initially, the adversary sets $\delta(v) = \bot$ for each vertex $v \in \gt{0}$,  it subdivides $\gt{0}$ to construct $\gt{1}$, and it sets $t=1$. By construction, invariants \ref{inv:defined1} and \ref{inv:defined2} are true.
Before the first query, $\config{A'}(1)$ is empty.
Since $\config{A}(1)$ is the set of all initial configurations and $\gt{0}$ represents all initial configurations,
invariant \ref{inv:defined3} is true.
 
No vertices in
$\gt{1}$ have terminated, so 
$\Tta{1}{a}$ is 
empty for all inputs $a$.
Since there have been no output queries, 
 $\Xta{1}{a}$ is
empty for all inputs $a$.
Therefore invariants \ref{inv:contains}, \ref{inv:outputsfar}, and \ref{inv:x} are vacuously true.

\medskip

Now suppose that the invariants are true immediately prior to some query $(C,q)$ in phase 1, where
$C \in \config{A}(1) \cup \config{A'}(1)$ and $q$ is a  process that is active in $C$.
Note that the prover already knows the state of every process in configuration $C$, including which of them
have terminated.
Let $\beta$ be a schedule from an initial configuration $C_0 \in \config{A}(1)$ such that $C = C_0 \beta$
and $C_0\beta' \in \config{A'}(1)$ for every nonempty prefix $\beta'$ of $\beta$.

If $q$ occurs $2r$ times in $\beta$, then, by invariant \ref{inv:defined3}, $0 \leq r \leq t$ and the state $s$ of $q$ in configuration $C$ is a vertex in $\gt{r}$. Since $q$ is active in $C$, $\delta(s) = \bot$.
Hence, by invariant \ref{inv:defined2}, $r < t$.
In this case, the adversary returns the configuration $Cq$, which is the same as $C$ except that 
$S_{r+1}[\mathit{id}(q)] = (\mathit{id}(q),s)$
and the state of $q$ 
has an extra bit indicating that it last performed an $\id{update}$.
Note that $q$ is active in this state.
Invariant \ref{inv:defined3} is true for configuration 
$Cq = C_0 \beta q$ since $p_i$ occurs $2r+1$ times in $\beta q$ and every other process is in the
same state in  configurations $C$ and $Cq$.
Since $\delta$ has not been changed by the adversary, 
$\tta{a}$ is unchanged for all inputs $a$ and invariants \ref{inv:defined1}, \ref{inv:defined2}
\ref{inv:contains}, and \ref{inv:outputsfar} remain true.
Since no vertices are added to $\xta{a}$ for any input $a$, invariant~\ref{inv:x} remains true.

So, suppose that $q$ occurs $2r+1$ times  in $\beta$. 
Let $\beta'$ be the longest prefix of $\beta$ in which $q$ occurs $2r$ times. Then $0 \leq r \leq t$ and the state $s$ of $q$ in configuration $C_0\beta'$ is a vertex in $\gt{r}$, by invariant \ref{inv:defined3}.
Since $q$ in active in $C$, it is active in configuration $C_0\beta'$, so $\delta(s) = \bot$. 
Hence, by invariant \ref{inv:defined2}, $r < t$.

The state $s'$ of $q$ in configuration $Cq$ is $(\mathit{id}(q),\sigma)$,
where $\sigma$ is the result of its $\id{scan}$ of $S_{r+1}$.
It is a vertex in $\gt{r+1}$.
Note that, by Observation~\ref{obs:configstates}, the contents of $S_{r+1}$ are determined by the states of all processes in $C$.
If $r < t-1$, then $\delta(s')$ is defined, by invariant  \ref{inv:defined1}. It is also possible that $r = t-1$ and $\delta(s')$ is defined.
In both these cases, the adversary returns configuration $Cq$, which is the same as $C$, except for the state of $q$ and,
if $\delta(s') \neq \bot$,
the value it outputs.
As above, all the invariants continue to hold.

Now, suppose that $r = t-1$ and $\delta(s')$ is not defined.
If there exists an input $a$ such that setting $\delta(s') = a$ maintains all the invariants,
then the adversary defines $\delta(s') = a$ and returns configuration $Cq$,
which is the same as $C$ except for the state of $q$ and the fact that $q$ outputs $a$.
In this case, the distance between $s'$ and $\nta$ is at least 2 and,
for all inputs $b\neq a$ such that $\tta{b}$ is nonempty, the distance between $s'$ and $\tta{b}$ in $\gt{t}$ is at least 3.
The vertex $s'$ is added to $\tta{a}$. The sets $\tta{b}$, for all inputs $b \neq a$, 
and the sets $\Nta{t}{b}$ and $\xta{b}$, for all inputs $b$, are unchanged.
Hence,  invariants~\ref{inv:defined1}, \ref{inv:defined2},
\ref{inv:contains}, \ref{inv:outputsfar},  and \ref{inv:x}
continue to hold. 
By construction, $s' \in \gt{t}$ and $\delta(s')$ is defined. For every other process, its state in $Cq$ is the same as its state in 
$C$. Thus, invariant~\ref{inv:defined3} continues to hold.
By invariant \ref{inv:x}, each vertex $u \in \xta{a}$ is either in $\nta$ or is at distance at most 1 from  $\tta{b}$ for some $b \neq a$.
Since the distance  between $s'$ and $\nta$ is at least 2 and the distance between $s'$ and $\tta{b}$\ is at least 3,
the distance between $s'$ and $u$ is at least 2. Thus $s' \not\in \xta{a}$, so defining $\delta(s') = a$
does not contradict the result of any previous output query.

Otherwise, the adversary defines $\delta(v) = \bot$
for each vertex $v \in \vt{t}$ where $\delta(v)$ is undefined, including $s'$,
subdivides $\gt{t}$ to construct $\gt{t+1}$, and increments $t$.
By Lemma~\ref{lem:invariants-hold2}, all the invariants continue to hold.
The adversary returns configuration $Cq$, which is the same as $C$ except for the state of $q$.

\medskip

Finally, suppose that the invariants are true immediately prior to some output query $(C,Q,y)$ in phase 1, where
$C \in \config{A}(1) \cup \config{A'}(1)$, each process $q \in Q$ is active in $C$, and $y$ is a possible output value.
Let $\mathbb{Q}$ be the set of vertices in $\gt{t}$ that represent the states of processes in $Q$ in
configurations reachable from $C$ via $Q$-only schedules.

If some vertex $v \in \mathbb{Q}$ has terminated with output $y$, then the adversary returns a $Q$-only schedule from $C$ that leads to a configuration $C'$ in which $v$ represents the state of a process in $C'$. 
None of the invariants are affected.

If every vertex in $\mathbb{Q}$ is in $\Nta{t}{y}$, $\xta{y}$, or $\tta{a}$, for some $a \neq y$, then it would be impossible for the adversary to return a $Q$-only  schedule from $C$ in which some vertex has terminated with output $y$ without violating validity or contradicting one of its previous answers.  In this case, the adversary adds $\mathbb{Q}$ to $\xta{y}$ and returns \none.
Note that adding vertices in $\Nta{t}{y}$ or $\tta{a}$ for $a \neq y$ does not make invariant~\ref{inv:x} false. 
The other invariants are not affected.

Otherwise, let $\mathbb{U} \neq \emptyset$ be the subset of vertices in $\mathbb{Q}$ that are not in
$\Nta{t}{y}$, $\xta{y}$, or $\tta{a}$, for some $a \neq y$.
For each vertex $u \in \mathbb{U}$, let $\mathbb{A}_u$ be the union of all $n$-vertex cliques in $\gt{t}$
containing $u$.
We consider three cases.

\textbf{Case 1}: \emph{There is a vertex $u \in  \mathbb{U}$ such that $\mathbb{A}_u \cap \tta{y}$ is nonempty.}
The adversary defines $\delta(v) = \bot$ for each vertex $v \in \vt{t}$ where $\delta(v)$ is undefined and subdivides $\gt{t}$ to construct $\gt{t+1}$. 
By invariant~\ref{inv:contains} and Lemma~\ref{lem:invariants-hold2}, the distance between $\Tta{t+1}{y}$ and $\Nta{t+1}{y}$ in $\gt{t+1}$ is at least 3.
If $a \neq y$ and $\tta{a}$ is nonempty, then Proposition~\ref{prop:tfacts} says that $\Tta{t+1}{a}$ is nonempty and,
by invariant~\ref{inv:outputsfar} and Lemma~\ref{lem:invariants-hold2},
the distance between $\Tta{t+1}{y}$ and $\Tta{t+1}{a}$ is at least 4.

Let $i  =  \mathit{id}(u)$ and  $v= (i, \{u\})$.
Since $u \in \mathbb{U} \subseteq \mathbb{Q}$, process $p_i \in Q$. 
Let 
$w \in \mathbb{A}_u \cap \tta{y}$, let $\sigma$ be an $n$-vertex clique in $\mathbb{A}_u$ that contains $w$, and
let $C'$ be the configuration represented by $\sigma$.
Then $v$ is the state of process $p_i$ in configuration $C'p_ip_i$.

Next, the adversary increments $t$, so all the invariants continue to hold
by Lemma~\ref{lem:invariants-hold2}.
Finally, the adversary defines $\delta(v) = y$,
returns a $Q$-only schedule from $C$ that results in process $p_i$  being in state $v$.
This adds vertex $v$ to $\tta{y}$.
Invariants  \ref{inv:defined1}, \ref{inv:defined2}, \ref{inv:defined3}, and  \ref{inv:x} continue to hold.

Since $w$ is terminated, it is adjacent to every other vertex in $\chi(\sigma,\delta) \subseteq \gt{t}$,
including $v$.
It follows that the distance between $v$ and $\Nta{t}{y}$ is at least 2 and,
if $a \neq y$ and $\tta{a}$ is nonempty, then the distance between $v$ and $\tta{a}$ is at least 3.
Thus, invariants \ref{inv:contains} and \ref{inv:outputsfar} hold.

By invariant \ref{inv:x}, each vertex in $\Xta{t}{y}$ is adjacent to a vertex in $\tta{b}$ for some $b \neq y$.
Since the distance  between $v$ and $\tta{b}$ is at least 3,
the distance between $v$ and $\Xta{t}{y}$ is at least 2. Thus $v \not\in \Xta{t}{y}$. Hence, defining $\delta(v) = y$
does not contradict the result of any previous output query.

\textbf{Case 2}:  \emph{There is a vertex $u \in  \mathbb{U}$ such that every vertex in $\mathbb{A}_u$ is active.}
The adversary defines $\delta(v) = \bot$ for each vertex $v \in \gt{t}$ where $\delta(v)$ is undefined and subdivides $\gt{t}$ to construct $\gt{t+1}$. 

Since no vertex in $\mathbb{A}_u$ has terminated and $\mathbb{A}_u$ contains all vertices at distance at most 1 from
$u$ in $\gt{t}$, it follows that the distance from $u$ to $\tta{a}$ in $\gt{t}$ is at least 2, for all inputs $a$.
Moreover, since $u \not\in \Nta{t}{y}$, the distance from $u$ to $\Nta{t}{y}$ in $\gt{t}$ is at least 1.

Let $i  =  \mathit{id}(u)$ and let $v= (i,\{ u \}) \in \gt{t+1}$.
Since $u \in \mathbb{U} \subseteq \mathbb{Q}$, process $p_i \in Q$.
Consider any vertex $v'$ adjacent to $v$ in $\gt{t+1}$. Then there exists an $n$-vertex clique $\sigma \subseteq \gt{t}$
such that $v,v' \in  \chi(\sigma,\delta)$.
Since $v$ is not a terminated vertex in $\sigma$, $\{u\} \subseteq  \sigma$,
so $\sigma \subseteq \mathbb{A}_u$. All vertices in $ \mathbb{A}_u$ are active, so $v' = (j,\rho)$ where 
$j \neq i$, $j \in \mathit{id}(\rho)$, 
$\rho \subseteq \sigma$, and  $\{u\} \subseteq \rho$.
Note that $u \not\in \Nta{t}{y}$ 
implies that 
$v,v' \not\in \Nta{t+1}{y}$. 
Therefore, the distance from $v$ to $\Nta{t+1}{y}$  in $\gt{t+1}$ is at least 2.

Next, we show that, for all inputs $a$, the distance from $v$ to $\Tta{t+1}{a}$  in $\gt{t+1}$ is at least 3.
By~\propositionref{prop:tfacts}, $\Tta{t+1}{a} = \tta{a}= \chi(\tta{a},\delta)$, so no vertex $v'$ adjacent to $v$
in $\gt{t+1}$ is in $\Tta{t+1}{a}$.
To obtain a contradiction, suppose there is 
 a path $v,v',w$ of length 2 in $\gt{t+1}$ from $v$ to  $\Tta{t+1}{a}$.
Then $v' = (j,\rho)$, where $\{u\} \subseteq \rho$  and 
there exists an $n$-vertex clique $\sigma' \subseteq \gt{t}$
such that $\{v',w\}$ is an edge in  $\chi(\sigma',\delta)$. Because $w \in \Tta{t+1}{a} = \tta{a}$, $w \in \sigma'$.
By definition, $\rho \subseteq \sigma'$.
This implies that $u \in \sigma'$ and, hence, $\sigma' \subseteq \mathbb{A}_u$. 
However, this contradicts the assumption that all vertices in $\mathbb{A}_u$ are active.
Therefore, the distance from $v$ to $\Tta{t+1}{a}$  in $\gt{t+1}$ is at least 3 for all inputs $a$.

Now the adversary increments $t$, so all the invariants continue to hold,
by Lemma~\ref{lem:invariants-hold2}.
Finally, the adversary defines $\delta(v) = y$ and
returns a $Q$-only schedule from $C$ that results in process $p_i$  being in state $v$.
This adds vertex $v$ to $\tta{y}$.
Invariants  \ref{inv:defined1}, \ref{inv:defined2}, \ref{inv:defined3}, and  \ref{inv:x} continue to hold.
Since the distance from $v$ to $\Nta{t}{y}$  in $\gt{t}$ is at least 2 and 
the distance from $v$ to $\Tta{t+1}{a}$  in $\gt{t+1}$ is at least 3 for all inputs $a$,
invariants \ref{inv:contains} and \ref{inv:outputsfar} hold.
As in the previous case,
defining $\delta(v) = y$ does not contradict the result of any previous output query.

\textbf{Case 3.} \emph{For every vertex $u \in \mathbb{U}$, $\mathbb{A}_u \cap \tta{y}$ is empty, but some vertex in
$\mathbb{A}_u$ has terminated.} 
In this case, the adversary returns $\none$ and adds $\mathbb{U}$ to $\Xta{t}{y}$.
Since each vertex $u \in \mathbb{U}$ is adjacent to some vertex in $\mathbb{A}_u$ that has terminated with an output other than $y$, invariant \ref{inv:x} holds.
Invariants  
\ref{inv:defined1}, \ref{inv:defined2}, and \ref{inv:defined3} still hold, since $t$ and $\delta$ are not changed,
and invariants \ref{inv:contains} and  \ref{inv:outputsfar} still hold,
since $\tta{a}$ and $\nta$ are not changed for any input $a$.

\paragraph*{The prover does not win in phase 1.} 
Suppose that the invariants all hold before and after each query made by the prover in phase 1. By invariant~\ref{inv:outputsfar}, at most one value is output in any configuration reached by the prover. Moreover, by invariant~\ref{inv:contains}, if a process outputs value $a$, then it has seen $a$.
Hence, the prover cannot win in phase 1 by showing that the protocol violates agreement or validity. It remains to  show that the prover cannot win by constructing an infinite chain of queries in phase 1.

\begin{lemma}
\label{lem:finitechains}
Every chain of queries in phase 1 is finite.
\end{lemma}
\begin{proof}
Assume, for a contradiction, that there is an infinite chain of queries, $(C_1,q_1),(C_2,q_2),\ldots$
Let $\beta$ be
a schedule
from an initial configuration $C_0$ to $C_1$ followed by the steps of the schedule
$q_1,q_2,\ldots$ and, for each $j \geq 1$, let $\beta_j$ be the prefix of $\beta$ such that $C_j = C_0 \beta_j$.
Let $P$ be the set of processes that occur infinitely often in $\beta$.
Let $j'  \geq 1$ be the
first
index such that $q_j \in P$ for all $j \geq j'$,
so, from $j'$ onwards, only processes in $P$ appear in queries.
Let $t' \geq 1$ be the value of $t$ held by the adversary immediately prior to query $(C_{j'},q_{j'})$. 
By invariant~\ref{inv:defined3}, each process occurs at most $2t+1$ times in $\beta_{j'}$.
Hence, during the schedule $\beta_{j'}$ from $C_0$, no process performed an \id{update} to $S_r$ for $r >t'$
or a \id{scan} of $S_r$ for $r \geq t'$.
Since each process in $P$ eventually accesses every snapshot object,
the adversary eventually
defines $\delta(v) = \bot$
for each vertex $v \in \gt{r}$ where $\delta(v)$ is undefined
and subdivides $\gt{r}$ to construct $\gt{r+1}$, for all $r \geq t'$.
Since no process is terminated, $\Tta{r}{a} = \Tta{t'}{a}$, for all inputs $a$ and all $r > t'$. 
By \lemmaref{lem:invariants-hold1} and \lemmaref{lem:invariants-hold2}, if $\Tta{t'}{a}$ is nonempty,
the distance between $\Tta{t'+2}{a}$ and $\Nta{t'+2}{a}$ is at least 4 and, if $b \neq a$ and $\Tta{t'}{b}$ is nonempty,
the distance between $\Tta{t'+2}{a}$ and $\Tta{t'+2}{b}$ is at least 5.
	
Consider the first index $j''  > j'$ such that process $q_{j''}$ is poised to \id{scan} the snapshot object $S_{t'+2}$ in $C_{j''}$.
By invariant~\ref{inv:defined3},
the  state of  process in $q_{j''}$ in configuration $C_{j''+1}= C{j''}q_{j''}$ 
is a vertex $v \in \vt{t'+2}$.
If there is some input $a$ such that the distance from $v$ to $\Tta{t'+2}{a}$ in $\gt{t'+2}$ is at most 2,
then the distance from $v$ to $\Nta{t'+2}{a}$  in $\gt{t'+2}$ is at least 2
and the distance from $v$ to $\Tta{t'+2}{b}$ in $\gt{t'+2}$ is at least 3.
According to its strategy for phase 1, the adversary defines $\delta(v) = a$ after query $(C_{j''},q_{j''})$. 
This contradicts the definition of $P$.
Thus, the distance from $v$ to  $\Tta{t+2}{a}$ in $\gt{t'+2}$ is at least 3, for all inputs $a$ such that
$\mathbb{T}_a^{t'+2}$ is nonempty.

Let $a$ be the input of process $q_{j''}$ in configuration $C_0$.
Consider any $(t'+2)$-round schedule $\beta''$
obtained from $\beta$ by removing all but the first $2(t'+2)$ occurrences of processes in $P$ and then appending 
sufficiently many occurrences of the processes not in $P$.
Note that configurations $C_0\beta''$ and $C_0\beta_{j''+1}$ are indistinguishable to process $q_{j''}$,
so $v$ is in the $n$-vertex clique $\sigma$ in $\gt{t'+2}$ representing the configuration $C_0\beta''$.
Thus the distance in $\gt{t'+2}$ between $\sigma$ and any terminated vertex is at least 2.
During schedule $\beta''$ from $C_0$, 
$q_{j''}$ performs its \id{update} to $S_{t'+2}$ before any process performs its \id{scan} of $S_{t'+2}$,
so all vertices in $\sigma$
have seen $a$.
Thus the distance in $\gt{t'+2}$ between $\sigma$ and $\Nta{t'+2}{a}$ is at least 1.
The first edge on every path from $\sigma$ to $\Nta{t'+2}{a}$ or to $\Tta{t'+2}{b}$, for any input $b$,
is between active vertices.
Therefore, by  \propositionref{prop:nfacts}, \propositionref{prop:tfacts}, and \lemmaref{lem:dist2}, the distance in  $\gt{t'+3}$ between $\chi(\sigma,\delta)$ and $\chi(\Nta{t'+2}{a},\delta) = \Nta{t'+3}{a}$ is at least 2
and the distance in  $\gt{t'+3}$ between $\chi(\sigma,\delta)$ and $\chi(\Tta{t'+2}{b},\delta) = \Tta{t'+3}{b}$ is at least 3,
for any input $b$.

Consider the first index $j'''  > j''$ such that process $q_{j'''}$ is poised to \id{scan} the snapshot object $S_{t'+3}$ in $C_{q'''}$. The state of process $q_{j'''}$ in configuration $C_{j'''}q_{j'''}$ is a vertex in $\chi(\sigma,\delta)$. According to its strategy for phase 1, the adversary terminates this vertex after query $(C_{j'''},q_{j'''})$. 
This contradicts the definition of $P$.
\end{proof}

Since the prover does not win in phase 1, it must eventually commit to a nonempty schedule $\alpha(2)$ from an initial
configuration $C \in \config{A}(1)$ such that $C\alpha(2) \in \config{A}'(1)$, set $\config{B}(2)$ to 
consist
of all initial configurations that only differ from $C$ by the states of processes that do not occur in $\alpha(2)$,
set $\config{A}(2) = \{C_0 \alpha(2) \ |\ C_0 \in \config{B}(2)\}$, and then start phase 2.

\paragraph{The adversarial strategy for later phases.}
At the beginning of phase 2, the adversary updates $\delta$. 
Afterwards, it can answer all future queries by the prover
without making any further changes to $\delta$.   Eventually, at the end of some future phase $\varphi$, the prover will  commit to a schedule $\alpha(\varphi+1)$ such that all configurations in $\config{A}(\varphi+1)$ are final.
Consequently, the prover will lose at the beginning of  phase $\varphi+1$.

Let $p$ be the first process in $\alpha(2)$ and let $a$ be the input of $p$ in the initial configuration $C$.
Note that $p$ has the same state in every configuration in $\config{B}(2)$, so it has input $a$ in all of them.
Let $\mathbb{F}$ denote the union of all $n$-vertex cliques in $\gt{1}$ that represent a configuration reachable by 
a 1-round schedule beginning with $p$ from a configuration in $\config{B}(2)$.
Since $p$ performs its \id{update} to $S_1$ before any process performs its \id{scan} of $S_1$ in all such schedules,
every vertex in $\mathbb{F}$ has seen $a$.
Thus the distance between $\mathbb{F}$ and $\Nta{1}{a}$ in $\gt{1}$ is at least 1.

The adversary defines $\delta(v) = \bot$ for each vertex $v \in \vt{t}$ where $\delta(v)$ is undefined, subdivides $\gt{t}$
to construct $\gt{t+1}$, and increments $t$.
Since all the invariants hold at the end of phase 1, \lemmaref{lem:invariants-hold2}
says that they still hold and,  for any two inputs $b \neq b'$ such that  $\tta{b}$ and $\tta{b'}$ are non-empty,
the distance between $\tta{b}$ and $\tta{b'}$ in $\gt{t}$ is at least 4.
In particular, a vertex $v \in \gt{t}$ is adjacent to a vertex $w \in \tta{b}$ for at most one input $b$.
Let $\mathbb{F}' = \chi^{t-1}(\mathbb{F},\delta) \subseteq \gt{t}$.
Applying \lemmaref{lem:invariants-hold2} $t-1$ times,
it follows that the distance between $\mathbb{F}'$ and $\Nta{t}{a}$ in $\gt{t}$ is at least 1.

Invariant~\ref{inv:defined2} says that no vertex in $\gt{t}$ has $\delta(v) = \bot$.
The adversary has not yet terminated any additional vertices in $\gt{t}$,
so, by \propositionref{prop:tfacts},
$\Tta{t}{b} = \Tta{t-1}{b}$ for all input values $b$.
For every vertex $v \in \mathbb{F}'$ for which $\delta(v)$ is undefined, the adversary defines $\delta(v)$ as follows.
First, for each input value $b$
and each vertex $v  \in \mathbb{F}'$  that is distance 1 from $\Tta{t-1}{b}$ in $\gt{t}$ and 
such that $\delta(v)$ is undefined,
the adversary sets $\delta(v) = b$. 
This does not violate validity, since the distance between $\Tta{t-1}{b}$ and $\Nta{t}{b}$ in $\gt{t}$ is at least 2.
Since each vertex in $\xta{b}$ is adjacent to a vertex in $\Tta{t-1}{b'}$ for some $b' \neq b$,
this assignment defines $\delta(v)$ for each vertex in $\xta{b}$ for which it was undefined.
Since each vertex in $\xta{b}$ is at least distance 3 from any vertex in $\Tta{t-1}{b}$,
this assignment does not contradict any output query that returned $\none$.
Moreover, the distance between any two vertices in $\mathbb{F}'$ that have output different values is still at least 2.
Thus, in each $n$-vertex simplex in $\gt{t}$, all the terminated vertices have output the same value. 

Finally,
for each vertex $v \in \mathbb{F}'$ where $\delta(v)$ is still undefined, the adversary sets $\delta(v) = a$. 
Validity is not violated, since no vertex in $\mathbb{F}'$ is in $\nta$.
Agreement is not violated, since at most two different values are output by the vertices
in each $n$-vertex simplex in $\gt{t}$.

In phases $\varphi \geq 2$, the prover can only query configurations reachable from some configuration in $\config{A}(2)$.
By definition, $\config{A}(2)$ is the set of all configurations that are reached by performing $\alpha(2)$ from initial configurations 
in $\config{B}(2)$.
It follows that, for any process $q$ and 
any extension $\alpha'$ of $\alpha(2)$ from $C' \in \config{A}(2)$, $q$ appears at most $2t$ times in $\alpha(2)\alpha'$ before its state is represented by a vertex in $\mathbb{F}'$. By construction, every vertex in $\mathbb{F}'$ has terminated. 
Thus, eventually, the prover chooses a configuration at the end of some phase in which every process has terminated.
The prover loses in the next phase.

 Thus, we have proved the following result:
 
 \begin{theorem}
No extension-based proof can show the
impossibility of deterministically solving $k$-set agreement in a wait-free manner in the NIS model, for $n > k \geq 2 $ processes.
 \end{theorem}
\section{Conclusions}
\label{sec:conclusions}

We have shown the limitation of extension-based proofs, including valency arguments, for proving the impossibility of deterministic, wait-free solutions to set-agreement in the NIS model.
In the conference version of this paper~\cite{AAEGZ19}, we obtained the same result in the NIIS model.
Although we have restricted
attention to the proof of impossibility of one problem in two closely related models,
our approach should be applicable to other problems and other models.
For example, we believe that there is no extension-based proof of the lower bound on the number of rounds to solve set agreement in synchronous message passing systems. 

Recently, Alistarh, Ellen, and Rybicki~\cite{AER20} proved that there is no extension-based proof of the impossibility
of deterministic, wait-free solutions to 4-cycle agreement for $n \geq 3$ processes in the NIIS model.
This result helped lead to their impossibility proof for this problem, which turned out to be similar to the impossibility proof for set agreement.

There are two other results in distributed computing that have a similar flavour.
Rincon Galeana, Winkler,  Schmid and Rajsbaum~\cite{GWSR19} showed that
partitioning arguments are insufficient to prove the impossibility of $(n-1)$-set agreement in the iterated immediate snapshot (IIS) model.
For the CONGEST model,
Bachrach, Censor-Hillel, Dory, Efron, Leitersdorf and Paz~\cite{BCDELP19} showed that reductions from two party communication complexity with a static cut cannot be used to prove non-constant lower bounds on the number of rounds needed to solve maximum matching or maximum flow.

Combinatorial topology has been used to prove the impossibility of wait-free solutions to problems other than set agreement, such as weak symmetry breaking and 
renaming~\cite{CR10}. There are no extension-based proofs of these results and we conjecture that they cannot be proved
using extension-based proofs.

\medskip

The definition of an extension-based proof can be modified to handle other termination conditions, such as obstruction-freedom \cite{herlihy2003obstruction}. It suffices for the prover to construct a schedule that violates this condition.

The NIS and NIIS models are computationally equivalent
to an asynchronous shared memory model in which processes communicate by reading from and writing to shared registers. However, these  models are not equivalent in terms of space and step complexities.
A covering argument~\cite{BL93} is a standard approach for proving a lower bound on the number of registers needed to solve a problem in an asynchronous system.  
We have a definition for extension-based proofs 
that includes
covering arguments. 

Ellen, Gelashvili and Zhu~\cite{EGZ18} proved that any obstruction-free protocol for $k$-set agreement
among  $n > k \geq 2$ processes
requires $\lceil n/k \rceil$ registers, but their proof is not extension-based. 
In fact, some of 
the
early work about extension-based proofs motivated the approach in~\cite{EGZ18}.
We conjecture that it is impossible to prove a non-constant lower bound on the number of registers needed by any obstruction-free protocol for $k$-set agreement using an extension-based proof.

We have considered allowing the prover to perform a number of other types of queries
and can extend our adversarial protocol so that it can answer them.
For example, if a prover asks the same output query
$(C,P,y)$ multiple times, the protocol could be required to return different schedules each time,
until it has returned all possible $P$-only schedules from $C$ that output $y$.

We cannot allow certain queries, such as asking for an upper bound on the length of any schedule.
If the prover is given such an upper bound,
then it can perform a finite number of chains of queries to examine all reachable configurations,
thereby fixing the protocol.
However, we can allow the prover to use this information in a restricted way
and still construct an adversarial set agreement protocol. For example, we might require that 
the prover does not use this information to decide which queries to perform or what extensions to construct, but can use this information to win when it has constructed a schedule that is longer than this upper bound.

\section{Acknowledgments}
Support is gratefully acknowledged from the 
Natural Science and Engineering Research Council of Canada under grants 
RGPIN-2015-05080 and
RGPIN-2020-04178,
a University of Toronto postdoctoral fellowship,
National Science Foundation under grants CCF-1217921, 
         CCF-1301926, CCF-1637385, CCF-1650596, and IIS-1447786, the Department of Energy under grant ER26116/DE-SC0008923, 
and the Oracle and Intel corporations.
We would also like to thank Toniann Pitassi for helpful discussions and Shi Hao Liu for his useful 
feedback
on an earlier draft of this paper.

\bibliographystyle{alpha}
\bibliography{new-biblio}
\end{document}